\documentclass[nofootinbib]{revtex4-2}

\makeatletter

\def\p@subsection{}

\def\p@subsubsection{}
\makeatother

\usepackage{amssymb}
\usepackage{amsthm}
\usepackage{mathtools} % coloneqq and friends
\usepackage{amsmath} % 'split','align' environment
\usepackage{physics} % \qty and \dyad (like \proj in phfqit)
\usepackage{mathtools} 
\usepackage{cancel} % \cancelto
\usepackage{bbm} % Allows double stroke 1: \mathbbm{1}
\usepackage{tabularx} % for protocol
\usepackage{wasysym}    % enumerate
\usepackage{enumitem}

\usepackage{subcaption}

\usepackage{graphicx}
\graphicspath{ {./images/} }

\newtheorem{theorem}{Theorem}
\newtheorem{lemma}[theorem]{Lemma}
\theoremstyle{remark}

\theoremstyle{definition}
\newtheorem{definition}[theorem]{Definition}

\newtheorem{corollary}[theorem]{Corollary}

\newcommand{\hmin}[1]{H^{#1}_{\mathrm{min}}}
\newcommand{\hmax}[1]{H^{#1}_{\mathrm{max}}}
\newcommand{\linearTermEAT}{t}
\renewcommand{\tr}{\mathrm{Tr}}

\usepackage{hyperref}
\hypersetup{
  colorlinks   = true,    % Colours links instead of boxes
  urlcolor     = black,    % Colour for external hyperlinks
  linkcolor    = black,    % Colour of internal links
  citecolor    = black      % Colour of citations
}

\newcounter{protocol}
\newenvironment{protocol}[1]
{
    \par\addvspace{\topsep}
    \noindent
    \tabularx{\linewidth}{@{} X @{}}
    \hline
    \refstepcounter{protocol}
    \textbf{Protocol \theprotocol} #1 \\
    \hline
}
{ \\
    \hline
    \endtabularx
    \par\addvspace{\topsep}
}

\begin{document}
\title{Entropy Accumulation under Post-Quantum Cryptographic Assumptions}
\author{Ilya Merkulov}
\author{Rotem Arnon-Friedman}
\affiliation{The Center for Quantum Science and Technology, Department of Physics of Complex Systems, Weizmann Institute of Science, Rehovot, Israel}

\date{\today}

%abstract
\begin{abstract}
    In device-independent (DI) quantum protocols, the security statements are oblivious to the characterization of the quantum apparatus -- they are based solely on the classical interaction with the quantum devices as well as some well-defined assumptions.
    The most commonly known setup is the so-called non-local one, in which two devices that cannot communicate between themselves present a violation of a Bell inequality. 
    In recent years, a new variant of DI protocols, that requires only a single device, arose.
    In this novel research avenue, the no-communication assumption is replaced with a computational assumption, namely, that the device cannot solve certain post-quantum cryptographic tasks.   
    The protocols for, e.g., randomness certification, in this setting that have been analyzed in the literature used ad hoc proof techniques and the strength of the achieved results is hard to judge and compare due to their complexity. 
    
    Here, we build on ideas coming from the study of non-local DI protocols and develop a modular proof technique for the single-device computational setting.
    We present a flexible framework for proving the security of such protocols by utilizing a combination of tools from quantum information theory, such as the entropic uncertainty relation and the entropy accumulation theorem.
    This leads to an insightful and simple proof of security, as well as to explicit quantitative bounds.   
    Our work acts as the basis for the analysis of future protocols for DI randomness generation, expansion, amplification and key distribution based on post-quantum cryptographic assumptions.  
\end{abstract}

\maketitle

    \section{Introduction}\label{sec:introduction}
    
The fields of quantum and post-quantum cryptography are rapidly evolving. 
In particular, the device-independent (DI) approach for quantum cryptography is being investigated in different setups and for various protocols. 
Consider a cryptographic protocol and a \emph{physical device} that is being used to implement the protocol. The DI paradigm states that when proving the security of the protocol, one should treat the device itself as an untrusted one, prepared by an adversarial entity.
Only limited, well-defined assumptions regarding the inner-workings of the device are placed. 
In such protocols, the honest party, called the verifier here, interacts with the untrusted device in a black-box manner, using classical communication. The security proofs are then based on properties of the  transcript of the interaction, i.e., the classical data collected during the execution of the protocol, and the underlying assumptions. 

The most well-studied DI setup is the so called ``non-local setting''~\cite{brunner2014bell,scarani2019bell}. There, the protocols are being implemented using (at least) \emph{two} untrusted devices and the assumption being made is that the devices cannot communicate between themselves during the execution of the protocol (or parts of it). 
In recent years, another variant has been introduced: Instead of working with two devices, the protocol requires only a \emph{single} device and the no-communication assumption is replaced by an assumption regarding the computational power of the device. More specifically, one assumes that during the execution of the protocol, the device is unable to solve certain \emph{computational problems}, such as Learning With Errors (LWE)~\cite{regev2009lattices}, which are believed to be hard for a quantum computer. (The exact setup and assumptions are explained in Section~\ref{sec:rand}).
Both models are DI, in the sense that the actions of the quantum devices are uncharacterized.
Figure~\ref{fig:setups} schematically presents the two scenarios.
%and both can be used for protocols that achieve ``ever-lasting secuirty'' statements.  

    \begin{figure}
        \includegraphics[scale=0.45]{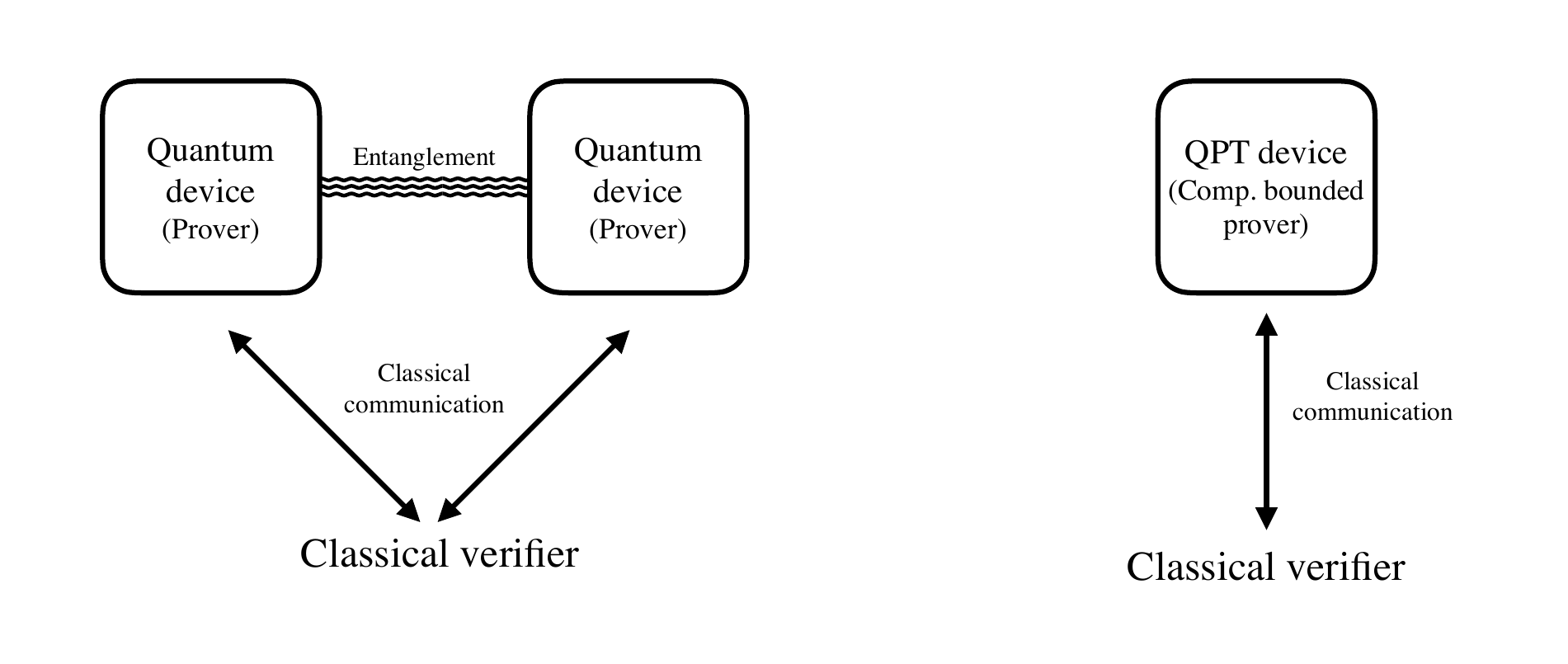}
    \centering
    \caption{Two setups for device-independent protocols.
    On the left, a classical verifier is interacting classically with two non-communicating but otherwise all powerful quantum devices (also called provers) that can share entanglement.
    On the right, the verifier is interacting with a single polynomial-time quantum computer.}
    \label{fig:setups}
    \end{figure}

As in practice it might be challenging to assure that two quantum devices do not communicate, as required in the non-local setting, the incentive to study what can be done using only a single device is high. 
Indeed, after the novel proposals made in~\cite{mahadev2018classical,brakerski2021cryptographic} for DI protocols for verification of computation and randomness certification, many more protocols for various tasks and computational assumptions were investigated; see, e.g.,~\cite{gheorghiu2019computationally,brakerski2020simpler,metger2021self,metger2021device,vidick2021classical,kahanamoku2022classically,liu2022depth,gheorghiu2022quantum,brakerski2023simple,natarajan2023bounding,aaronson2023certified}.  
Experimental works also followed, aiming at verifying quantum computation in the model of a single computationally restricted device~\cite{zhu2021interactive,stricker2022towards}.
With this research agenda advancing, more theoretical and experimental progress needs to be done-- a necessity before one could estimate whether this avenue is of relevance for future quantum technologies or of sheer theoretical interest. 

In this work we are interested in the task of generating randomness.\footnote{Unless otherwise written, when we discuss the generation of randomness we refer to a broad family of tasks: randomness certification, expansion and amplification as well as, potentially, quantum key distribution. In the context of the current work, the differences are minor.} 
We consider a situation in which the device is prepared by a quantum adversary and then given to the verifier (the end-user or costumer). The verifier wishes to use the device in order to produce a sequence of random bits. In particular, the bits should be random also from the perspective of the quantum adversary. 
The protocol that indicates the form in which the verifier interacts with the device (interchangeably also called the prover) should be constructed in a way that guarantees that either the verifier aborts with high probability if the device is not to be trusted or that, indeed, we can certify that the bits produced by the device are random and unknown to the adversary. 

Let us slightly formalize the above. 
The story begins with the adversary preparing a quantum state $\rho^{\text{in}}_{PE}$; the marginal $\rho^{\text{in}}_{P}$ is the initial state of the device given to the verifier while the adversary keeps the quantum register $E$ for herself. 
In addition to the initial state $\rho^{\text{in}}_{P}$, the device is described by the quantum operations, e.g., measurements, that it performs during the execution of the protocol. 
We assume that the \emph{device} is a quantum polynomial time (QPT) device and thus can only perform efficient operations. Namely, the initial state $\rho^{\text{in}}_{P}$ is of polynomial size and the operations can be described by a polynomial size quantum circuit; these are formally defined in Section~\ref{sec:pre}. 
The verifier can perform only (efficient) classical operations.
Together with the verifier, the initial state of all parties is denoted by $\rho^{\text{in}}_{PVE}$.\footnote{In a simple scenario one can consider $\rho^{\text{in}}_{PVE}=\rho^{\text{in}}_{PE}\otimes\rho^{\text{in}}_{V}$, i.e., the verifier is initially decoupled from the device and the adversary. This is, however, not necessarily the case in, e.g., randomness amplification protocols~\cite{kessler2020device}. We therefore allow for this flexibility with the above more general notation.} 
The verifier then executes the considered protocol using the device.

All protocols consist of what we call ``test rounds'' and ``generation rounds''. 
The goal of a test round is to allow the verifier to check that the device is performing the operations that it is asked to apply by making it pass a test that only certain quantum devices can pass. The test and its correctness are based on the chosen computational assumption.\footnote{In the non-local setting, the test is based on the violation of a Bell inequality, or winning a non-local game with sufficiently high winning probability. Here the computational assumption ``replaces'' the Bell inequality.} 
For example, in~\cite{brakerski2021cryptographic}, the cryptographic scheme being used is ``Trapdoor Claw-Free Functions'' (TCF)-- a family of pairs of injective functions~$f_{0},f_{1}:\{0,1\}^n \rightarrow \{0,1\}^n$.
Informally speaking, it is assumed that, for every image $y$: (a) Given a ``trapdoor'' one could classically and efficiently compute two pre-images~$x_0 ,x_1$ such that $f_0(x_0)=f_1(x_1)=y$; (b) Without a trapdoor, even a quantum computer cannot come up with~$x_0 ,x_1$ such that $f_0(x_0)=f_1(x_1)=y$ (with high probability). 
While there exists no efficient quantum algorithm that can compute both pre-images~$x_0, x_1$ for a given image~$y$ without a trapdoor, a quantum device can nonetheless hold a \emph{superposition} of the pre-images by computing the function over a uniform superposition of all inputs to receive~$\sum_{y\in\{0,1\}^n}\left(\ket{0,x_0} + \ket{1,x_1}\right)\ket{y}$.
These insights (and more-- see Section~\ref{sec:pre_post_quantum} for details) allow one to define a test based on TCF such that a quantum device that creates $\sum_{y\in\{0,1\}^n}\left(\ket{0,x_0} + \ket{1,x_1}\right)\ket{y}$ can win while other devices that do not hold a trapdoor cannot. Moreover, the verifier, holding the trapdoor, will be able to check classically that the device indeed passes the test.

Let us move on to the generation rounds. In these rounds, the device produces the \emph{output bits} $O$, which are supposed to be random.
During the execution of the protocol some additional information, such as a chosen public key for example (or whatever is determined by the protocol), may be publicly announced or leaked; we denote this \emph{side information} by $S$. 

After executing all the test and generation rounds the verifier checks if the average winning probability in the test rounds is higher than some  pre-determined threshold probability $\omega\in(0,1)$. If it is, then the protocol continues and otherwise aborts.
Let the final state of the entire system, conditioned on not aborting the protocol, be~$\rho_{|\Omega}$. 
To show that randomness has been produced, one needs to lower-bound the conditional smooth min-entropy~$H^{\varepsilon}_{\min}(O|SE)_{\rho_{|\Omega}}$~\cite{tomamichel2010duality} (the formal definitions are given in Section~\ref{sec:pre}). Indeed, this is the quantity that tightly describes the amount of \emph{information-theoretically} uniform bits that can be extracted from the output~$O$, given~$S$ and~$E$, using a quantum-proof randomness extractor~\cite{renner2005universally,de2012trevisan}. 
The focus of our work is to supply explicit lower-bounds on $H^{\varepsilon}_{\min}(O|SE)_{\rho_{|\Omega}}$ in a general and modular way. 

An important remark is in order before continuing. Notice that in the above the \emph{device} is computationally bounded and yet we ask to get output bits $O$ that are information-theoretically secure with respect to the \emph{adversary}. This means that the adversary may keep her system $E$ and perform on it, using the knowledge of $S$, \emph{any}, not necessarily efficient, operation. By proving such a strong statement, it is implied that the computational assumption only needs to hold during the execution of the protocol in order to provide ever-lasting security of the final output. This also means that the outcomes are random with respect to the adversary even if the computational assumption, e.g., hardness of LWE, is broken after the creation of the bits. This property is termed ``security lifting'' and is a fundamental and crucial feature of all DI protocols based on computational assumptions.

\subsection{Motivation of the current work}

    The setup of two non-communicating devices has naturally emerged from the study of quantum key distribution~(QKD) protocols and non-local games (or Bell inequalities). As such, the quantum information theoretic toolkit for proving the security of protocols such as DIQKD, DI randomness certification and alike was developed over many years and used for the analysis of numerous quantum protocols (see, e.g., the survey~\cite{primaatmaja2022security}). 
    The well-established toolkit includes powerful techniques that allow bounding the conditional smooth min-entropy $H^{\varepsilon}_{\min}(O|SE)_{\rho_{|\Omega}}$ mentioned above. Examples for such tools are the entropic uncertainty relations~\cite{berta2010uncertainty}, the entropy accumulation theorem~\cite{dupuis2020entropy,metger2022generalised} and more. 
    The usage of these tools allows one to derive quantitatively strong lower bounds on $H^{\varepsilon}_{\min}(O|SE)_{\rho_{|\Omega}}$ as well as handle realistically noisy quantum devices~\cite{arnon2018practical,zapatero2023advances}.

    Unlike the protocols in the non-local setup, the newly developed protocols for, e.g., randomness certification with a single device restricted by its computational power, are each analyzed using ad hoc proof techniques. 
    On the qualitative side, such proofs make it harder to separate the wheat from the chaff, resulting in a less modular and insightful claims. 
    Quantitatively, the strength of the achieved statements is hard to judge-- they are most likely not strong enough to lead to practical applications and it is unknown whether this is due to a fundamental difficulty or a result of the proof technique. As a consequence, it is unclear whether such protocols are of relevance for  future technology. 

    In this work, we show how to combine the information theoretic toolkit with assumptions regarding the computational power of the device. More specifically, we prove lower bounds on the quantity $H^{\varepsilon}_{\min}(O|SE)_{\rho_{|\Omega}}$ by exploiting post-quantum cryptographic assumptions and quantum information-theoretic techniques, in particular the entropic uncertainty relation and the entropy accumulation theorem. Prior to our work, it was believed that such an approach cannot be taken in the computational setting (see the discussion in~\cite{brakerski2021cryptographic}). 
    Once a bound on $H^{\varepsilon}_{\min}(O|SE)_{\rho_{|\Omega}}$ is proven, the security of the considered protocols then follows from our bounds using standard tools.
    The developed framework is general and modular. We use the original work of~\cite{brakerski2021cryptographic} as an explicit example; the same steps can be easily applied to, e.g., the protocols studied in the recent works~\cite{brakerski2023simple,natarajan2023bounding}.

\subsection{Main ideas and results}

    The main tool which is used to lower bound the conditional smooth min-entropy in DI protocols in the non-local setting is the entropy accumulation theorem (EAT)~\cite{dupuis2020entropy,arnon2019simple}. The EAT deals with sequential protocols, namely, protocols that proceed in rounds, one after the other, and in each round some bits are being output. Roughly speaking, the theorem allows us to relate the total amount of entropy that accumulates throughout the execution of the protocol to, in some sense, an ``average worst-case entropy of a single round'' (see Section~\ref{sec:pre} for the formal statements). 
    It was previously unclear how to use the EAT in the context of computational assumptions and so~\cite{brakerski2021cryptographic} used an ad hoc proof to bound the total entropy. 

    The general structure of a protocol that we consider is shown in Figure~\ref{fig:eat}. 
    The initial state of the system is $\rho^{\text{in}}_{PVE}$. The protocol, as mentioned, proceeds in rounds. Each round includes interaction between the verifier and the prover and, overall, can be described by an efficient quantum channel $\mathcal{M}_i$ for round $i\in[n]$. The channels output some outcomes~$O_i$ and side information $S_i$. In addition, the device (as well as the verifier) may keep quantum and classical memory from previous rounds-- this is denoted in the figure by the registers $R_i$. We remark that there is only one device and the figure merely describes the way that the protocol proceeds. That is, in each round $i\in[n]$ the combination of the actions of the device and the verifier, together, define the maps $\mathcal{M}_i$. The adversary's system $E$ is untouched by the protocol.\footnote{One could also consider more complex protocols in which the adversary's information does change in some ways. This can be covered using the generalized EAT~\cite{metger2022generalised}. We do not do so in the current manuscript since all protocols so far fall in the above description but the results can be extended to the setup of~\cite{metger2022generalised}.}

    \begin{figure}
        \includegraphics[scale=0.5]{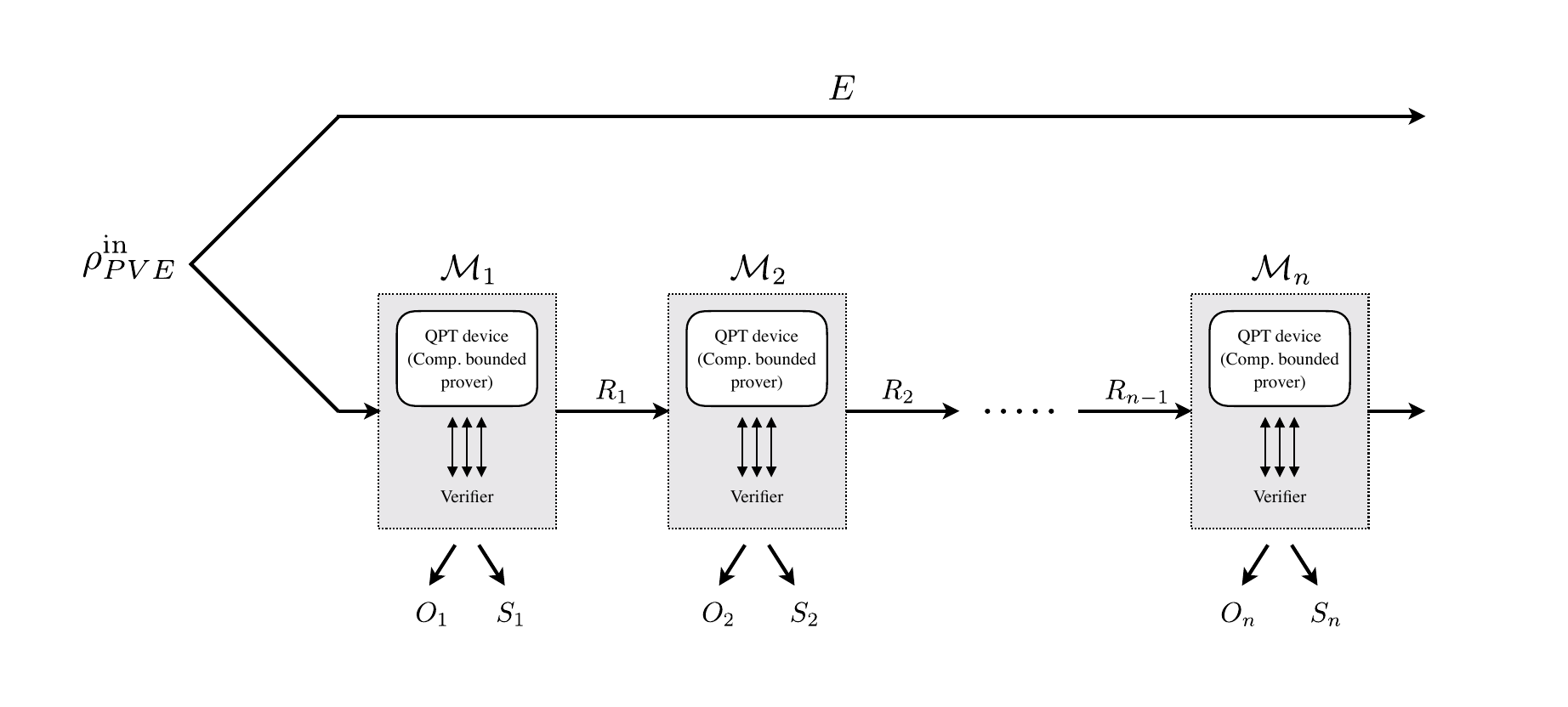}
    \centering
    \caption{The general structure of a protocol that we consider. The initial state of the entire system is $\rho^{in}_{PVE}$. The protocol proceeds in rounds: Each round includes interaction between the verifier and the prover, as shown by the gray boxes in the figure, and can be described by an efficient quantum channel $\mathcal{M}_i$ for every round $i\in[n]$. The channels output outcomes $O_i$ and side information $S_i$. The device may keep quantum memory from previous rounds using the registers $R_i$. The adversary’s system $E$ is untouched by the protocol. This structure fits the setup of the entropy accumulation theorem~\cite{dupuis2020entropy} and can easily be extended to that of the generalized entropy accumulation~\cite{metger2022generalised}. }\label{fig:eat}
    \end{figure}
    
    We are interested in bounding the entropy that accumulated by the end of the protocol: $H^{\varepsilon}_{\min}(\mathbf{O}|\mathbf{S}E)_{\rho_{|\Omega}}$, where $\mathbf{O}=O_1,\cdots,O_n$ and similarly $\mathbf{S}=S_1,\cdots,S_n$. 
    The EAT tells us that, under certain conditions, this quantity is lower bounded, to first order in $n$, by $\linearTermEAT n$ with $\linearTermEAT$  of the form 
    \begin{equation}\label{eq:t_intro}
        \linearTermEAT = \min_{\sigma\in \Sigma(\omega)} H(O_i|S_iE)_{\mathcal{M}_i(\sigma)} \;,
    \end{equation}
    where $\omega$ is the average winning probability of the device in the test rounds and $\Sigma(\omega)$ is the set of \emph{all} (including inefficiently constructed) states of polynomial size that achieve the winning probability $\omega$. 

    To lower bound the entropy appearing in Equation~\eqref{eq:t_intro} one needs to use the computational assumption being made. This can be done in various ways. 
    In the current work, we show how statements about anti-commuting operators, such as those proven in~\cite{brakerski2021cryptographic,brakerski2023simple,natarajan2023bounding}, can be brought together with the entropic uncertainty relation~\cite{berta2010uncertainty}, another tool frequently used in quantum information theory, to get the desired bound. 
    In combination with our usage of the EAT the bound on $H^{\varepsilon}_{\min}(\mathbf{O}|\mathbf{S}E)_{\rho_{|\Omega}}$ follows.

    The described techniques are being made formal in the rest of the manuscript. We use~\cite{brakerski2021cryptographic} as an explicit example,\footnote{For readers who are familiar with~\cite{brakerski2021cryptographic}, our work exploits a few lemmas from~\cite[Sections~6 and~7]{brakerski2021cryptographic}, which define ``the heart'' of the computational assumption in the context of randomness generation and then we replace~\cite[Section~8]{brakerski2021cryptographic} completely.} also deriving quantitative bounds. We discuss the implications of the quantitative results in Section~\ref{sec:conc} and their importance for future works. 
    Furthermore, our results can be used directly to prove full security of the DIQKD protocol of~\cite{metger2021device}.

    We add that on top of the generality and modularity of our technique, its simplicity contributes to a better understanding of the usage of post-quantum cryptographic assumptions in randomness generation protocols. 
    Thus, apart from being a tool for the analysis of protocols, we also shed light on what is required for a protocol to be strong and useful.
    For example, we can see from the explanation given in this section that the computational assumptions enter in three forms in Equation~\eqref{eq:t_intro}:
    \begin{enumerate}
        \item The channels $\mathcal{M}_i$ must be efficient.
        \item The states $\sigma$ must be of polynomial size.
        \item Due to the minimization, the states $\sigma$ that we need to consider may also be inefficient to construct, even though the device is efficient.
    \end{enumerate}
    Points 1 and 2 are the basis when considering the computational assumption and constructing the test; In particular, they allow one to bound $H(O_i|S_iE)_{\mathcal{M}_i(\sigma)}$, up to a negligible function $\eta(\lambda)$, where~$\lambda$ is a security parameter defined by the computational assumption. 
    
    Point 3 is slightly different. The set over which the minimization is taken, determines the strength of the computational assumption that one needs to consider. For example, it indicates that the protocol in~\cite{brakerski2021cryptographic} requires that the LWE assumption is hard even with a potentially-inefficient polynomial-size advice state. This is a (potentially) stronger assumption than the ``standard'' formulation of LWE. 
    Note that the stronger assumption is required even though the initial state of the device, $\rho^{\text{in}}_{P}$ \emph{is} efficient to prepare (i.e., it does not act as an advice state in this case). 
    In~\cite{brakerski2021cryptographic}, this delicate issue arises only when analyzing everything that can happen to the initial state throughout the entire protocol and conditioning on not aborting.
    In the current work, we directly see and deal with the need of allowing advice states from the minimization in Equation~\eqref{eq:t_intro}.

    \section{Preliminaries}\label{sec:pre}
    Throughout this work, we use the symbol~$\mathbbm{1}$ as the identity operator and as the characteristic function interchangeably;  the usage is clear due to context.
We denote~$x\leftarrow X$ when~$x$ is sampled uniformly from the set~$X$ or~$x\leftarrow\mathcal{D}$ when~$x$ is sampled according to a distribution~$\mathcal{D}$.
The Bernoulli distribution with~$p(0)=\gamma$ is denoted as~$\text{Bernoulli}(1-\gamma)$
The Pauli operators  are denoted by~$\sigma_{x},\sigma_{y},\sigma_{z}$.
The set~$\qty{1,\cdots,n}$ is denoted as~$\qty[n]$.

\subsection{Mathematical background}
    \begin{definition}[Negligible function]
        A function~$\eta:\mathbb{N}\rightarrow\mathbb{R}^{+}$ is to be negligible if for every~$c \in \mathbb{N}$ there exists~$N\in\mathbb{N}$ such that for every~$n>N$,~$\eta(n)<n^{-c}$.
    \end{definition}

    \begin{lemma}
        Let $\eta:\mathbb{N}\rightarrow\mathbb{R}^{+}$ be a negligible function. The function $\eta\qty(n)\ln\qty(1/\eta\qty(n))$ is also negligible.
    \end{lemma}
    \begin{proof}
        Assume, w.l.o.g, monotonicity and positivity of some negligible function~$\eta$.
        Given $c\in\mathbb{N}$, we wish to find $N\in\mathbb{N}$ s.t. $\forall n>N$, $\eta\qty(n)\ln\qty(1/\eta\qty(n)) < n^{-c}$.
        We denote the following function
        \begin{equation*}
            g\qty(n) \coloneqq \max{\qty{\left.k\in\mathbb{N}\right|\eta\qty(n)<n^{-k}}} \; .
        \end{equation*}
        This means
        \begin{equation*}
            \begin{array}{ccccccl}
                & n^{-(g(n)+1)} & \leq & \eta(n) & < & n^{-g(n)} & , \\
                \Rightarrow & -(g(n)+1)\ln{n} & \leq & \ln{\eta(n)} & < & -g(n)\ln{n} & , \\
                \Rightarrow & (g(n)+1)\ln{n} & \geq & \ln{\qty(1/\eta(n))} & > & g(n)\ln{n} & ,
            \end{array}
        \end{equation*}
        yielding
        \begin{equation*}\eta\qty(n)\ln\qty(1/\eta\qty(n)) <
            n^{-g(n)}(g(n)+1)\ln{n} =
            \frac{g(n)+1}{n^{g(n)-1}}\frac{\ln{n}}{n} <
            \frac{g(n)+1}{n^{g(n)-1}} =
            \frac{g(n)+1}{n^{g(n)-1-c}}\frac{1}{n^c} \; .
        \end{equation*}
        By the negligibility property of $\eta(n)$, the function $g(n)$ diverges to infinity.
        Therefore, $\exists N\in\mathbb{N}$ s.t. $\forall n > N$
        \begin{equation*}
            \frac{g(n)+1}{n^{g(n)-1-c}} \leq 1\;,
        \end{equation*}
        and for all~$n>N$
        \begin{equation*}\eta\qty(n)\ln\qty(1/\eta\qty(n)) <
            \frac{g(n)+1}{n^{g(n)-1-c}}\frac{1}{n^c} \leq 
            \frac{1}{n^c}\;. \qedhere
        \end{equation*}
    \end{proof}
    \begin{corollary}
        Let~$\eta$ be a negligible function and let~$h(x)=-x\log x - (1-x)\log(1-x)$ be the binary entropy function. There exists a negligible function~$\xi$ for which the following holds.
        \begin{equation}\label{eq:binary_entropy_egligible_function_arg}
            h(x - \eta(n)) \geq h(x) - \xi(n) \;.
        \end{equation}
    \end{corollary}
    
    \begin{definition}[Hellinger distance]\label{def:hellinger}
        Given two probability distributions~$P=\qty{p_i}_{i}$ and $Q=\qty{q_i}_{i}$, the Hellinger distance between~$P$ and~$Q$ is defined as
        \begin{equation*}\label{eq:hellinger}
            H(P,Q) = \frac{1}{\sqrt{2}}\sqrt{\sum_{i} \qty(\sqrt{p_i} + \sqrt{q_i})^2}\;.
        \end{equation*}
    \end{definition}

    \begin{lemma}[Jordan's lemma extension]\label{lemma:jordans_extension}
        Let~$U_1,U_2$ be two self adjoint unitary operators acting on a Hilbert space~$\mathcal{H}$, of a countable dimension.
        Let~$L$ be a normal operator acting on the same space such that~
        $\qty[L,U_{1}]=\qty[L,U_{2}]=0$.
        There exists a decomposition of the Hilbert space into a direct sum of orthogonal subspaces~$\mathcal{H} = \oplus_{\alpha} \mathcal{H}_{\alpha}$
        % \begin{equation}\label{eq:hilbert_decomposition}
        %     \mathcal{H} = \oplus_{\alpha} \mathcal{H}_{\alpha} 
        % \end{equation}
        such that for all~$\alpha$, $\dim{\qty(\mathcal{H}_{\alpha})} \leq 2$ and given~$\ket{\psi} \in
        \mathcal{H}_\alpha$ all 3 operators satisfy~
        $U_{1}\ket{\psi},U_{2}\ket{\psi},L\ket{\psi}\in\mathcal{H}_{\alpha}$.
    \end{lemma}

    Jordan's lemma appears, among other places, in~\cite[Appendix G.4]{scarani2019bell}.
    Lemma~
    The sole change in the proof of the extension is in the choice of diagonalizing basis of the unitary operator~$U_2 U_1$.
    The chosen basis is now a mutual diagonalizing basis of~$U_2 U_1$ and~$L$ which exists due to commutative relations.
    % \begin{proof}
    %     We repeat the same steps of lemma~\ref{lemma:jordans} but instead of choosing an arbitrary orthonormal diagonalizing basis of~$U_{2}U_{1}$ we choose one that is also a diagonalizing basis of~$L$ which exists since
    %     both~$U_1$ and~$U_2$ commute with~$L$.
    
    %     One then sees that~$\ket{\psi}\in\text{Span}\qty(\qty{\ket{\alpha},\ket{\tilde{\alpha}}})$, as defined in
    %     lemma~\ref{lemma:jordans}, is also closed to application of~$L$ since both~$\ket{\alpha},\ket{\tilde{\alpha}}$ are eigen vectors of~$L$.
    %     The first one by definition of the basis and the second is evident by
    %     \begin{equation*}
    %         L\ket{\tilde{\alpha}} =
    %         L U_{2}\ket{\alpha} =
    %         U_{2} L \ket{\alpha} =
    %         \lambda U_{2} \ket{\alpha} =
    %         \lambda \ket{\tilde{\alpha}}
    %     \end{equation*}
    
    %     where~$\lambda$ is the complex eigenvalue of~$\ket{\alpha}$ with respect to~$L$.
    % \end{proof}

    \begin{lemma}\label{lemma:jordans_projections}
        Let~$\Pi, M, K$ be Hermitian projections acting on a Hilbert space~$\mathcal{H}$ of a countable dimension such
        that~$[K,\Pi]=[K,M]=0$.
        There exists a decomposition of the Hilbert space into a direct sum of orthogonal subspaces such that~$\Pi,M$ and~$K$ are 2
        by 2 block diagonal.
        In addition, in subspaces~$\mathcal{H}_{\alpha}$ of dimension 2,~$\Pi$ and $M$ take the forms
        \begin{equation*}
            \begin{array}{ccc}
                \Pi_{\alpha}=\left(
                    \begin{array}{cc}
                        1     \\
                        \phantom{0}  & 0
                    \end{array}
                \right)
    
                & &
    
                M_{\alpha}=\left(
                    \begin{array}{cc}
                        c_{\alpha}^{2}       & c_{\alpha}s_{\alpha}\\
                        c_{\alpha}s_{\alpha} & s_{\alpha}^{2}
                    \end{array}
                \right)\end{array}
        \end{equation*}
        where~$c_\alpha = \cos{\theta_\alpha},s_\alpha=\sin{\theta_\alpha}$ for some~$\theta_{\alpha}$.
    \end{lemma}
    \begin{proof}
        Given a Hermitian projection~$\Pi$, Consider the unitary operator,
        ~$2\Pi-\mathbbm{1}$, satisfying~$\qty(2\Pi-\mathbbm{1})^{2} = \mathbbm{1}$.
        Therefore, the operators~$2\Pi-\mathbbm{1}, 2M-\mathbbm{1}$ satisfy the conditions for Lemma~\ref{lemma:jordans_extension} and there exists decomposition of~$\mathcal{H}$ to a direct sum of orthogonal subspaces~$\mathcal{H}_{\alpha}$ such that in these subspaces
        \begin{equation*}
            \begin{array}{ccc}
                2\Pi_{\alpha} - \mathbbm{1}_{\alpha}=\left(
                    \begin{array}{cc}
                            & 1\\
                        1 &
                    \end{array}
                \right)
    
                & &
    
                2M_{\alpha} - \mathbbm{1}_{\alpha}=\left(
                    \begin{array}{cc}
                                        & \omega \\
                        \bar{\omega} &
                    \end{array}
                \right)\end{array} \;.
        \end{equation*}
        One can recognize the operators as~$\sigma_x$ and $ \cos{\theta}\sigma_x + \sin{\theta}\sigma_y$, respectively, for
        some angle~$\theta$.
        Therefore, there exists a basis of~$\mathcal{H}_{\alpha}$ such that the operators are~$\sigma_z$ and $
        \cos{\phi}\sigma_z + \sin{\phi}\sigma_x$, respectively, for some angle~$\phi$.
        Consequently, in this subspace, the projections take the form
        \begin{equation*}
            \begin{array}{ccc}
                \Pi_{\alpha}=\frac{1}{2}\qty(
                    \qty(
                        \begin{array}{cc}
                            1 & \\
                                & -1
                        \end{array}
                    )
                    +
                    \qty(
                        \begin{array}{cc}
                            1 & \\
                                & 1
                        \end{array}
                    )
                )
                =
                \qty(
                    \begin{array}{cc}
                        1 & \\
                            & 0
                    \end{array}
                )
                \\ \\
                M_{\alpha}=\frac{1}{2}\qty(
                    \qty(
                        \begin{array}{cc}
                            \cos{\phi} & \sin{\phi}\\
                            \sin{\phi }& -\cos{\phi}
                        \end{array}
                    )
                    +
                    \qty(
                        \begin{array}{cc}
                            1 & \\
                                & 1
                        \end{array}
                    )
                )
                =
                \qty(
                    \begin{array}{cc}
                        c_{\alpha}^2         & c_{\alpha}s_{\alpha}\\
                        c_{\alpha}s_{\alpha} & c_{\alpha}^2
                    \end{array}
                )
                \end{array} \;,
        \end{equation*}
        with~$\theta_{\alpha} = \phi/2$.
    \end{proof}

    \begin{lemma}[Jensen's inequality extension]\label{lemma:jensen_extension}
        Let~$f:U\rightarrow\mathbb{R}$ be a convex function over a convex set~$U\in\mathbb{R}^{n}$ such that for every~$u\in U$ there exists a subgradient.
        Let~$\qty(X_i)_{i\in[n]}$ be a sequence of random variables with support over~$U$.
        Then, $\mathbb{E}\qty[f\qty(X_1, ..., X_n)] \geq f\qty(\mathbb{E}\qty[X_1],...,\mathbb{E}\qty[X_n])$.
    \end{lemma}
    
\subsection{Tools in quantum information theory}

We state here the main quantum information theoretic definitions and techniques appearing in previous work, on which we build in the current manuscript. 

    \begin{definition}[von Neumann entropy]
        Let~$\rho_{AB}$ be a density over the Hilbert space~$\mathcal{H}_A \otimes \mathcal{H}_B $. The von Neumann entropy of the ~$\rho_A$ is defined as
        \begin{equation*}
            H(A)_{\rho} = - \tr(\rho_A \log \rho_A) \; .
        \end{equation*}
        The conditional von Neumann entropy of~$A$ given~$B$ is defined to be
        \begin{equation*}
            H(A|B)_\rho = H(AB)_\rho - H(B)_\rho \; .
        \end{equation*}
    \end{definition}
    
    \begin{definition}[Square overlap~{\cite{berta2010uncertainty}}]\label{def:square_overlap}
        Let~$\Pi$ and~$M$ be two observables, described by orthonormal bases~$\qty{\ket{\pi_i}}_i$ and~$\qty{\ket{m_j}}_j$ on a~$d$-dimensional Hilbert space~$\mathcal{H}_A$.
        The measurement processes are then described by the completely positive maps
        \begin{equation}\label{eq:measurement_map}
            \mathcal{P}:\rho\rightarrow\sum_{i} \bra{\pi_i} \rho \ket{\pi_i} \dyad{\pi_i}
            \qquad , \qquad
            \mathcal{M}:\rho\rightarrow\sum_{j} \bra{m_j} \rho \ket{m_j} \dyad{m_j} .
        \end{equation}
        The square overlap of~$\Pi$ and~$M$ is then defined as
        \begin{equation*}
            c\coloneqq \max_{i,j}
                {
                    \left|\braket{\pi_i}{m_j}\right|^2
                } 
            \; .
        \end{equation*}
    \end{definition}

    \begin{definition}
        Let~$M$ be an observable on a~$d$ dimensional Hilbert space~$\mathcal{H}_A$ and let~$\mathcal{M}$ be its corresponding map as appearing in Equation\eqref{eq:measurement_map}.
        Given a state~$\rho\in\mathcal{H}_A\otimes\mathcal{H}_B$, we define the conditional entropy of the measurement~$M$ given the side information~$B$ as
        \begin{equation}\label{eq:observable_entropy}
            H(M|B)_\rho = H(A|B)_{(\mathcal{M}\otimes\mathbbm{1}_B)(\rho)} = H\qty(AB)_{(\mathcal{M}\otimes\mathbbm{1}_B) (\rho)} - H\qty(B)_{(\mathcal{M}\otimes\mathbbm{1}_B)(\rho)} \; .
        \end{equation}
    \end{definition}

    \begin{lemma}[Entropic uncertainty principle~{\cite[Supplementary -- Corollary 2]{berta2010uncertainty}}]
        \label{lemma:uncertainty_relations}
        Let~$\Pi$ and~$M$ be two observables on a Hilbert space~$\mathcal{H}_P$ and let~$c$ be their square overlap.
        For any density operator~$\rho\in\mathcal{H}_V\otimes\mathcal{H}_P\otimes\mathcal{H}_E$,
        \begin{equation*}
            H(\Pi|E) \geq \log (1/c) - H(M|V) \; .
        \end{equation*}
    \end{lemma}

    In order to provide a clear understanding of the quantum uncertainty that arises from two measurements, as in Lemma~\ref{lemma:uncertainty_relations}, it is beneficial to examine the square overlap between those measurements (Definition~\ref{def:square_overlap}).
    The Bloch sphere representation, which pertains to Hilbert spaces of two dimensions, offers a lucid illustration of this concept, as shown in Figure~\ref{fig:uncertainty_bloch}.

    \begin{figure}
        \centering
        \begin{subfigure}[b]{0.4\textwidth}
            \centering
            \includegraphics[scale=0.15]{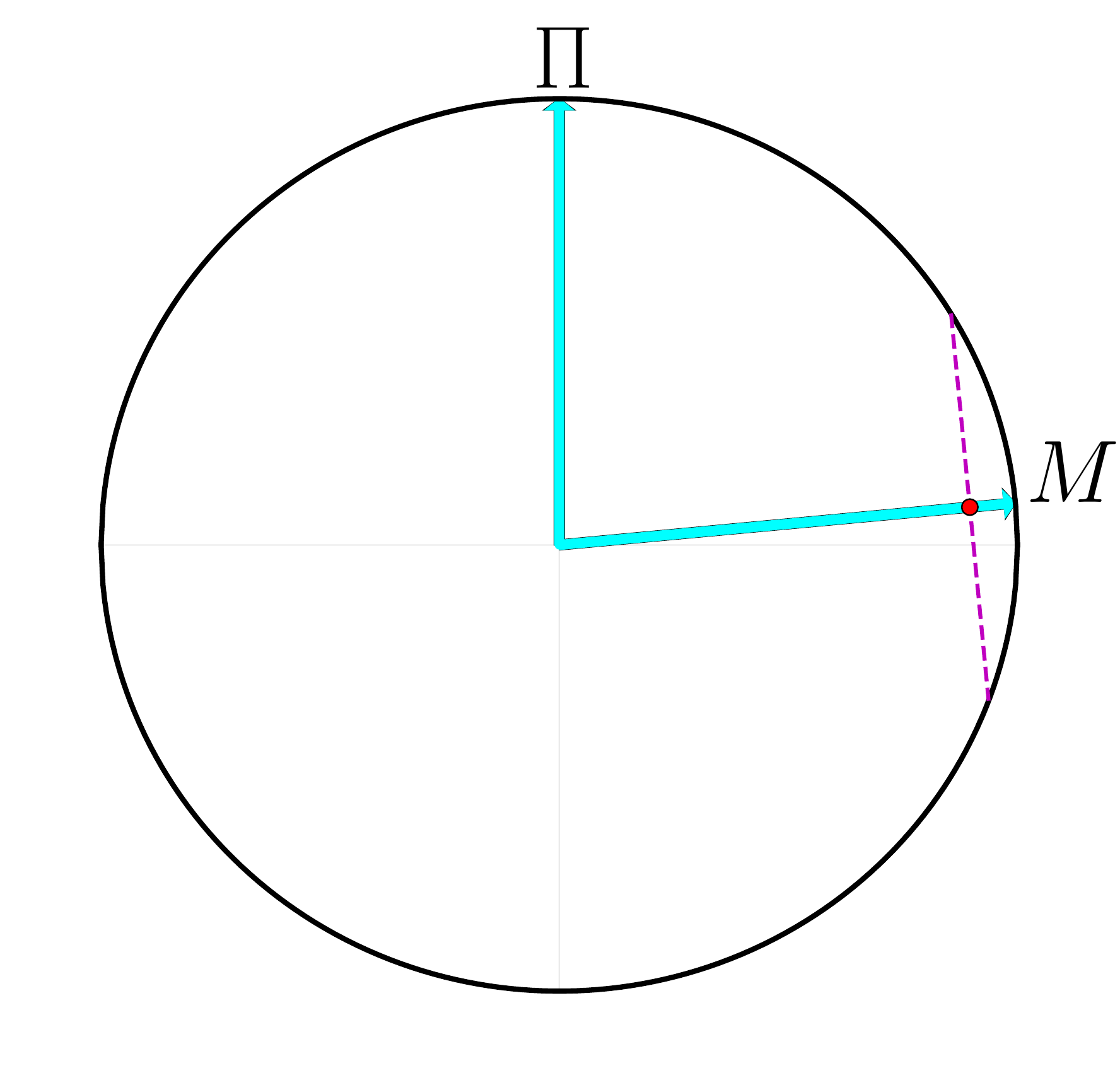}
            \caption
            {
                All possible states with~$\Pr(M=0)$ on the magenta dashed line.
            }
            \label{sub_fig:a}
        \end{subfigure}
        \begin{subfigure}[b]{0.4\textwidth}
            \centering
            \includegraphics[scale=0.15]{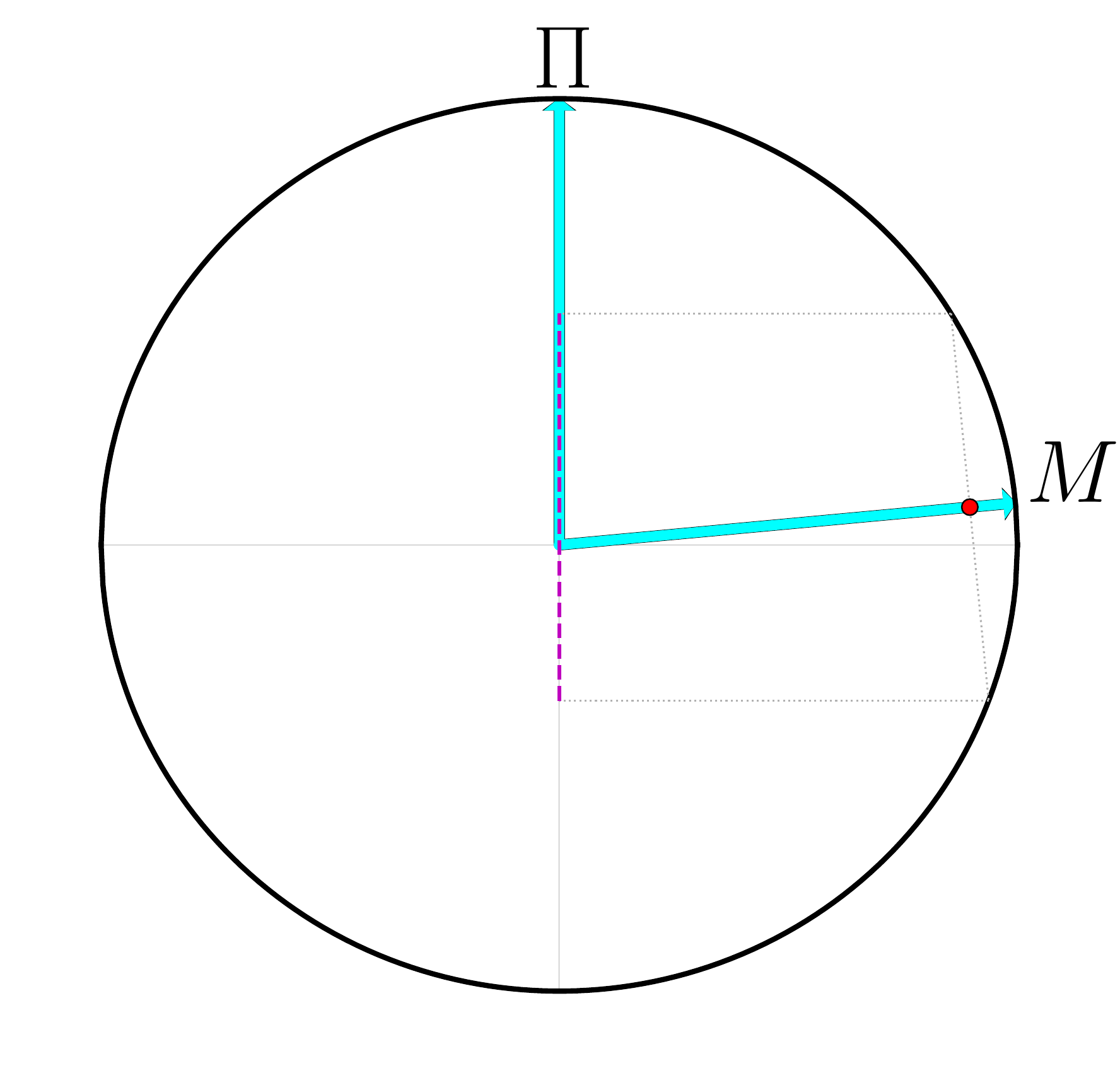}
            \caption
            {
                All possible distributions of~$\Pr(\Pi=0)$ on the magenta dashed line.
            }
            \label{sub_fig:b}
        \end{subfigure}
        \caption
        {
            A Bloch sphere representation of the uncertainty relations.
            A small square overlap between~$\Pi$ and~$M$ corresponds to a larger angle between the two.
            Each point on the arrow~$M$ corresponds to a probability~$\Pr(M=0)$ of some unknown state~$\rho$ that must lie on the corresponding magenta dashed line in (a).
            All the states described in (a)  allow distributions of~$\Pr(\Pi=0)$ that correspond to the points on the dashed magenta line described in (b). Small square overlap and high values of~$\Pr(M=0)$ mean that the allowed distributions of~$\Pr(\Pi=0)$ are close to uniform.
        }
        \label{fig:uncertainty_bloch}
    \end{figure}

    \begin{lemma}\label{claim:eigenvalue_lemma}
        Given two non-trivial Hermitian projections
        \begin{equation}\label{eq:qubit_projections}
            \Pi=\frac{1}{2}\qty(\mathbbm{1}+\sigma_{z})~~;~~\\
            M=\frac{1}{2}\qty(\mathbbm{1}+\cos{\qty(\theta)}\sigma_{z}+\sin{\qty(\theta)}\sigma_{x})
        \end{equation}
        acting on a 2-dimensional Hilbert space~$\mathcal{H}$, the eigenvalues of the operator $\Pi M \Pi + (\mathbbm{1} - \Pi) M (\mathbbm{1} - \Pi)$ are $\cos^{2}{\qty(\theta/2)}$ and $\sin^{2}{\qty(\theta/2)}$.
        %  \begin{equation*}\label{eq:projection_sum_operator}
        %     \Pi M \Pi + (\mathbbm{1} - \Pi) M (\mathbbm{1} - \Pi)
        % \end{equation*}
        % are
        % \begin{equation*}\label{eq:projection_sum_operator_ev}
        %     \cos^{2}{\qty(\theta/2)}~~;~~
        %     \sin^{2}{\qty(\theta/2)} \; .
        % \end{equation*}
        In addition, the square overlap of the operators in Equation~\eqref{eq:qubit_projections} is~$c=\max\qty{\cos^2{(\theta/2)},\sin^2{(\theta/2)}}$.
    \end{lemma}
   
    \begin{lemma}[Continuity of conditional entropies~{\cite[Lemma~2]{winter2016continuity}}]
        \label{lemma:entropy_continuity}
        For states~$\rho$ and~$\sigma$ on a Hilbert space~$\mathcal{H}_A\otimes\mathcal{H}_B$, if~$\frac{1}{2}\norm{\rho-\sigma}_1 \leq \epsilon \leq 1$, then
        \begin{equation}\label{eq:entropy_continuity}
        |H(A|B)_{\rho} - H(A|B)_{\sigma}| \leq
            \epsilon \log_2 (|A|) + (1+\epsilon) h\qty(\frac{\epsilon}{1+\epsilon}) \; .
        \end{equation}
    \end{lemma}

    \begin{lemma}[Good-subspace projection{~\cite[Lemma 7.2]{brakerski2021cryptographic}}]\label{lemma:good_subspace_projection}
        Let~$\Pi$ and~$M$ be two Hermitian projections on~$\mathcal{H}$ and~$\phi$ a state on~$\mathcal{H}$.
        Let~$\omega = \Tr(M\phi)$ and
        \begin{equation*}
            \mu = \qty
                |
                    \frac{1}{2} - \Tr(M\Pi\phi\Pi) - \Tr(M(\mathbbm{1}-\Pi)\phi(\mathbbm{1}-\Pi))
                | \;.
        \end{equation*}
        Let~$c\in(1/2,1]$.
        Let~$B_j$ be the orthogonal projection on the~$j$'th~$2\times2$ block, as given in Lemma~\ref{lemma:jordans_projections} and denote~$c_j$ as the square overlap of~$\Pi$ and~$M$ in the corresponding  block.
        Define~$\Gamma$ be the orthogonal projection on all blocks such that the square overlap is bound by c:
        \begin{equation}
            \Gamma \coloneqq \sum_{j:c_j \leq c} B_j \; .
        \end{equation}
        Then,
        \begin{equation}
            \tr((\mathbbm{1}-\Gamma)\phi) \leq \frac{2\mu + 10\sqrt{1 - \omega}}{{(2c - 1)}^2} \; .
        \end{equation}
    \end{lemma}
    The above lemma was proven in~\cite{brakerski2021cryptographic}.
    Since in plays a crucial part in the work we include the proof for completeness.
    \begin{proof}
        Using Jordan's lemma we find a basis of~$\mathcal{H}$ in which
        \begin{equation}
            M = \oplus_j
            \begin{pmatrix}
                a_j^2 & a_j b_j \\
                a_j b_j & b_j ^2
            \end{pmatrix}
            ~~~~ \text{and} ~~~~
            \Pi = \oplus_j
            \begin{pmatrix}
                1 &   \\
                  & 0
            \end{pmatrix}
        \end{equation}
        where~$a_j = \cos \theta_j$, $b_j = \sin \theta_j$, for some angles~$\theta_j$.
        Let~$\Gamma$ be the orthogonal projection on those 2-dimensional blocks such that~$\max(a_j^2, b_j^2)\equiv c_j \leq c$. 
        Note that: 
        \begin{enumerate}
            \item $\Gamma$ commutes with both~$M$ and~$\Pi$, but not necessarily with~$\phi$.
            \item  $\Gamma$ is the projection on~$2\times2$ blocks where the eigenvalues of the operator $\Pi M\Pi + (\mathbbm{1}-\Pi)M(\mathbbm{1}-\Pi)$ are in the range~$[1-c,c]$ as a consequence of Lemma~\ref{claim:eigenvalue_lemma}.
        \end{enumerate}

        Assume that~$\omega = 1$.
        This implies that~$\phi$ is supported on the range of~$M$.
        For any block~$j$, let~$B_j$ be the projection on the block and~$p_j = \Tr(B_j \phi)$.
        It follows from the decomposition of~$M$ and~$\Pi$ to 2 dimensional blocks and the definition of~$\mu$ that\footnote{This is easily seen in the Bloch sphere representation.}
        \begin{equation}\label{eq:mu_lower}
            \qty|
                \sum_{j:c_j \leq c} p_j \qty(a_j^4 + b_j^4) + 
                \sum_{j:c_j > c} p_j \qty(a_j^4 + b_j^4) - 
                \frac{1}{2}
                | = \mu \; .
        \end{equation}
        Using that for~$j$ with~$c_j > c$ such that~$\max(a_j^2, b_j^2) \geq c$ and the fact that the polynomial $1-2x(1-x)$ is strictly rising in the regime~$x>1/2$, we get
        \begin{equation*}
            a_j^4 + b_j^4 = 1 - 2 a_j^2  b_j^2 =
            1 - 2 \max(a_j^2, b_j^2)(1 - \max(a_j^2, b_j^2))
            \geq 1 - 2c(1-c) =
            \frac{1}{2} + \frac{1}{2}\qty(2c-1)^2 \; .
        \end{equation*}
        From Equation~\eqref{eq:mu_lower} and~$a_j^4 + b_j^4 \geq 1/2$ it follows that
        \begin{align*}
        2 \mu & = 
        \sum_{j:c_j \leq c} p_j \qty(1+(2c_j-1)^2) + 
        \sum_{j:c_j > c} p_j \qty(1+(2c_j-1)^2) - 1 \\
        & = \sum_{j:c_j \leq c} p_j \qty(2c_j-1)^2 + \sum_{j:c_j > c} p_j \qty(2c_j-1)^2
        \geq \sum_{j:c_j > c} p_j \qty(2c-1)^2\; .
        \end{align*}
        Consequently, 
        \begin{equation}\label{eq:no_omega_proj}
            \tr((\mathbbm{1}-\Gamma)\phi) \leq \frac{2\mu}{{(2c - 1)}^2} \; .
        \end{equation}
        This concludes the case where~$\omega=1$.
        
        Consider~$\omega < 1$ and~$\Tr(M\phi) > 0$, as otherwise the lemma is trivial.
        Let~$\phi' = M\phi M/ \Tr(M \phi)$.
        By the gentle measurement lemma~\cite[Lemma 9.4.1]{wilde2013gentle},
        \begin{equation}\label{eq:gentle_projection}
            \norm{\phi' -\phi}_1 \leq 2\sqrt{1 - \omega} \; .
        \end{equation}
        Using the definition of~$\mu$ it follows that
        \begin{equation*}
            \qty
                |
                    \frac{1}{2} - \Tr(M\Pi \phi' \Pi) - \Tr(M(\mathbbm{1}-\Pi)\phi'(\mathbbm{1}-\Pi))
                |
            \leq
                \mu + 4\sqrt{1 - \omega}
                \; .
        \end{equation*}
        Following the same steps as those used for the case~$\gamma = 0$ yields an analogue of~\eqref{eq:no_omega_proj}, with~$\phi'$ instead of~$\phi$ on the left-hand side and~$\mu+4\sqrt{1 - \omega}$ instead of~$\mu$ on the right-hand side.
        Applying again Equation~\eqref{eq:gentle_projection} the same bound transfers to~$\phi$ with an additional loss of~$2\sqrt{1 - \omega}$.
    \end{proof}
.
    A main tool that we are going to use is the entropy accumulation theorem (EAT)~\cite{dupuis2020entropy,arnon2018practical,dupuis2019entropy}. 
    For our quantitive results we used the version appearing in~\cite{dupuis2020entropy}; One can use any other version of the EAT in order to optimize the randomness rates, e.g.,~\cite{dupuis2019entropy} as well as the generalization in~\cite{metger2022generalised}.
    We do not explain the EAT in detail; the interested reader is directed to~\cite{arnon2020device} for a pedagogical explanation.
        
    \begin{definition}[EAT channels]\label{def:eat_channels} 
        Quantum channels~$\qty{\mathcal{M}_i :R_{i-1}\rightarrow R_i O_i S_i Q_i}_{i\in[n]}$ are said to be EAT channels if the following requirements hold:
        \begin{enumerate}
            \item $\qty{O_i}_{i\in[n]}$ are finite dimensional quantum systems of dimension~$d_O$ and~$\qty{Q_i}_{i\in[n]}$ are finite-dimensional classical systems (RV).~$\qty{S_i}_{i\in[n]}$ are arbitrary quantum systems.
            \item For any~$i\in[n]$ and any input state~$\sigma_{R_{i-1}}$, the output state~$\sigma_{R_i O_i S_i}=\mathcal{M}_{i} (\sigma_{R_{i-1}})$ has the property that the classical value~$Q_i$ can be measured from the marginal~$\sigma_{O_i S_i}$ without changing the state. That is, for the map~$\mathcal{T}_i : O_i S_i \rightarrow O_i S_i Q_i$ describing the process of deriving~$Q_i$ from~$O_i$ and~$S_i$, it holds that~$\Tr_{Q_i}\circ\mathcal{T}_i (\sigma_{O_i S_i})=\sigma_{O_i S_i}$.
            \item For any initial state~$\rho^{\text{in}}_{R_0 E}$, the final state~$\rho_{\mathbf{OSX}E}=(\Tr_{R_n}\circ\mathcal{M}_n \circ\cdots\circ\mathcal{M}_1)\otimes\mathbbm{1}_E \rho^{\text{in}}_{R_0 E}$ fulfils the Markov-chain conditions $ O_1,...,O_{i-1} \leftrightarrow S_1,...,S_{i-1},E\leftrightarrow S_i$.
            % \begin{equation*}\label{eq:markov_chain_condition}
            %     O_1,...,O_{i-1} \leftrightarrow S_1,...,S_{i-1},E\leftrightarrow S_i
            % \end{equation*}
        \end{enumerate}
    \end{definition}

    \begin{definition}[Min-tradeoff function]
        \label{def:min_tradeoff_function}
        Let ~$\qty{\mathcal{M}}_{i}$ be a family of EAT channels and~$\mathcal{Q}$ denote the common alphabet of~$Q_1 ,..., Q_n$. A differentiable and convex function~$f_{\text{min}}$ from the set of probability distributions~$p$ over~$\mathcal{Q}$ to the real numbers is called a~\emph{min-tradeoff function} for~$\qty{\mathcal{M}}_i$ if it satisfies
        \begin{equation*}
            f_{\text{min}} (p) \leq 
            \inf_{\sigma_{R_{i-1} R'} : \mathcal{M}_i (\sigma)_{Q_i}=p} 
            {H(O_i | S_i R')}_{\mathcal{M}_i (\sigma)}
        \end{equation*}
        for all~$i\in[n]$, where the infimum is taken over all purifications of input states of~$\mathcal{M}_i$ for which the marginal on~$Q_i$ of the output state is the probability distribution~$p$.
        % Furthermore, we define the variance of such function as
        % \begin{equation*}\label{eq:function_variance}
        %     \text{Var}(f) \coloneqq
        %     \max_{p:\Sigma(p)\neq \emptyset}
        %     \sum_{q\in\mathcal{Q}}p(q)f(\delta_{q})^2
        %     - \qty(\sum_{q\in\mathcal{Q}}p(q)f(\delta_{q}))^2
        % \end{equation*}
        % with~$\delta_q$ being the distribution with all weight over~$q\in\mathcal{Q}$ and~$\Sigma(p) = \qty{\sigma_{R_i R'} : \mathcal{M}_i (\sigma)_{Q_i} = p}$.
    \end{definition}

    \begin{definition}
        Given an alphabet~$\mathcal{Q}$, for any $\mathbf{q}\in\mathcal{Q}^n$ we define the probability distribution~$\mathrm{freq}_{\mathbf{q}}$ over~$\mathcal{Q}$ such that for any~$\tilde{q}\in\mathcal{Q}$,
        \begin{equation*}
            \mathrm{freq}_{\mathbf{q}}(\tilde{q}) = \frac{\qty|\qty{i\in[n]:q_i=\tilde{q}}|}{n} \;.
        \end{equation*}
    \end{definition}

    \begin{theorem}
    [Entropy Accumulation Theorem (EAT)]
    \label{theorem:eat}
        Let~$\mathcal{M}_i : R_{i-1} \rightarrow R_i O_i S_i Q_i$ for~$i\in[n]$ be EAT channels,
        let~$\rho$ be the final state,
        $\Omega$ an event defined over~$\mathcal{Q}^n$,
        $p_\Omega$ the probability of~$\Omega$ in~$\rho$ and~$\rho_{|\Omega}$ the final state conditioned on~$\Omega$.
        Let~$\varepsilon\in(0,1)$.
        For~$\hat{\Omega}=\qty{\mathrm{freq}_{\mathbf{q}} : \mathbf{q}\in\Omega}$ convex,
        $f_\mathrm{min}$ a min-tradeoff function for~$\qty{\mathcal{M}_i}_{i\in\qty[n]}$,
        and any~$\linearTermEAT\in\mathbb{R}$ such that~$f_{\mathrm{min}}(\mathrm{freq}_\mathbf{q}) \geq \linearTermEAT$ for any~$\mathrm{freq}_{\mathbf{q}}\in\hat{\Omega}$,
        \begin{equation}\label{eq:eat}
            H^{\varepsilon}_{\mathrm{min}} \qty(\mathbf{O} | \mathbf{S} E)_{\rho_{|\Omega}} \geq n \linearTermEAT - \mu \sqrt{n} \; ,
        \end{equation}
        where
        \begin{equation}\label{eq:second_order_term}
            \mu = 2 (\log (1+2 d_O) + \left\lceil \left\Vert \nabla f_{\mathrm{max}}\right\Vert _{\infty}\right\rceil  ) \sqrt{1 - 2\log(\varepsilon\cdot p_{\Omega})} \; .
        \end{equation}
    \end{theorem}
    
    % Dupuis&Fawzi version. We eventually did not use it so it remains as a comment. AF is used instead.
    % \begin{theorem}[Entropy Accumulation Theorem (EAT)~\cite{dupuis2019entropy}]
    %     Let~$\mathcal{M}_i : R_{i-1} \rightarrow R_i O_i S_i$ for~$i\in[n]$ be EAT channels,
    %     let~$\rho_{\mathbf{O B Q} E}$ be the final state,
    %     let~$h\in\mathbb{R}$,
    %     let~$f$ be an \emph{affine} min-tradeoff function for the aforementioned EAT channels, let~$\varepsilon\in(0,1)$, and assume~$O$ is a classical outcome.
    %     Then, for any event~$\Omega \subseteq \mathcal{Q}^{n}$ that implies~$f(\mathrm{freq}(Q_{1}^{n})) \geq h$ and for any~$\alpha\in(1,2)$
    %     \begin{equation}
    %         \hmin{\varepsilon} (\mathbf{O}|\mathbf{S}E)_{\rho_{|\Omega}} > nh - n\frac{(\alpha-1)\ln 2}{2} V^2 - \frac{1}{\alpha-1}\log\frac{2}{\varepsilon^{2} p_{\Omega}^2} - n\qty(\alpha-1)^2 K_{\alpha}
    %     \end{equation}
    %     holds for
    %     \begin{align*}
    %         V & = \sqrt{\mathrm{Var}(f)+2} + \log(2d_{O}^{2}+1),\\
    %         K_{\alpha} & = \frac{1}{6(2-\alpha)^3 \ln 2}\cdot
    %         2^{(\alpha-1)(\log d_O + (\mathrm{Max}(f)-\mathrm{Min}(f))}
    %         \ln^{3} (2^{(\log d_O + (\mathrm{Max}(f)-\mathrm{Min}(f))} + e^2) \; ,
    %     \end{align*}
    %     where the state~$\rho$ conditioned on the event~$\Omega\subseteq\mathcal{Q}$ is given by
    %     \begin{equation*}
    %         \rho_{\mathbf{Q O S} E|\Omega} =
    %             \frac{1}{p_\Omega}
    %             \sum_{\mathbf{q}\in\Omega}
    %             \dyad{\mathbf{q}}
    %             \otimes
    %             \rho_{\mathbf{A B E, q}} \; .
    %     \end{equation*}
    % \end{theorem}
    
\subsection{Post-quantum cryptography}\label{sec:pre_post_quantum}
    \newcommand{\KF}{\mathcal{K}_{\mathcal{F}}}
    \begin{definition}[Trapdoor claw-free function family]
        For every security parameter~$\lambda\in\mathbb{N}$, Let~$\mathcal{X}\subseteq\qty{0,1}^{w},\mathcal{Y}$ and~$\KF$ be finite sets of inputs, outputs and keys respectively.
        A family of injective functions 
        \begin{equation*}\label{eq:tcf_family}
            \mathcal{F}=\qty{f_{k,b}:\mathcal{X}\rightarrow \mathcal{Y}}_{k\in\KF,b\in\qty{0,1}}
        \end{equation*}
        is said to be trapdoor claw-free (TCF) family if the following holds:
        \begin{itemize}
            \item \textbf{Efficient Function Generation:} There exists a PPT algorithm~$\text{Gen}_\mathcal{F}$ which takes the security parameter~$1^\lambda$ and outputs a key~$k\in\KF$ and a trapdoor~$t$.
            \item \textbf{Trapdoor:} For all keys~$k\in\KF$ there exists an efficient deterministic algorithm~$\text{Inv}_{k}$ such that, given~$t$, for all~$b\in\qty{0,1}$ and~$x\in\mathcal{X}$, $\text{Inv}_{k}(t,b,f_{k,b}(x))=x$.
            \item \textbf{Claw-Free:} For every QPT algorithm~$\mathcal{A}$ receiving as input~$(1^{\lambda},k)$ and outputting a pair~$(x_0,x_1)\in\mathcal{X}^2$ the probability to find a claw is negligible.
            I.e. there exists a negligible function~$\eta$ for which the following holds.
            \begin{equation*}
                \Pr_{\overset{k\leftarrow\KF}{(x_0,x_1)\leftarrow\mathcal{A}(1^\lambda,k)}}
                \qty[f_{k,0}(x_0)=f_{k,1}(x_1)] \leq \eta(\lambda).
            \end{equation*}
        \end{itemize}
    \end{definition}

    \begin{definition}[Noisy NTCF family~\cite{brakerski2021cryptographic}]\label{def:ntcf}
        For every security parameter~$\lambda\in\mathbb{N}$, Let~$\mathcal{X},\mathcal{Y}$ and~$\KF$ be finite sets of inputs, outputs and keys respectively and~$\mathcal{D}_\mathcal{Y}$ the set of distributions over~$\mathcal{Y}$. A family of functions 
        \begin{equation*}\label{eq:ntcf_family}
            \mathcal{F}=\qty{f_{k,b}:\mathcal{X}\rightarrow \mathcal{D}_{Y}}_{k\in\KF,b\in\qty{0,1}}
        \end{equation*}
        is said to be noisy trapdoor claw-free (NTCF) family if the following conditions hold:
        \begin{itemize}
            \item \textbf{Efficient Function Generation:} There exists a probabilistic polynomial time (PPT) algorithm~$\text{Gen}_\mathcal{F}$ which takes the security parameter~$1^\lambda$ and outputs a key~$k\in\KF$ and a trapdoor~$t$.
            \item \textbf{Trapdoor Injective Pair:} For all keys~$k\in\KF$, the following 2 conditions are satisfied.
            \begin{enumerate}
                \item Trapdoor: For all~$b\in\qty{0,1}$ and~$x\neq x'\in\mathcal{X},~\text{Supp}(f_{k,b}(x))\cap\text{Supp}(f_{k,b}(x'))=\emptyset$. In addition, there exists an efficient deterministic algorithm~$\text{Inv}_{k}$ such that for all~$b\in\qty{0,1},~x\in\mathcal{X}$ and~$y\in\text{Supp}(f_{k,b} (x))$, $\text{Inv}_{k}(t,b,y)=x$.
                \item Injective Pair: There exists a perfect matching relation~$\mathcal{R}_{k}\subseteq\mathcal{X}\times\mathcal{X}$ such that~$f_{k,0}(x_0)=f_{k,1}(x_1)$ if and only if~$(x_0,x_1)\in\mathcal{R}_k$.
            \end{enumerate}
            \item \textbf{Efficient Range Superposition:} For every function in the family~$f_{k,b}\in\mathcal{F}$, there exists a function~$f'_{k,b}:\mathcal{X}\rightarrow\mathcal{D}_\mathcal{Y}$ (not necessarily a member of~$\mathcal{F}$) such that the following hold.
            \begin{enumerate}
                \item For all~$(x_0,x_1)\in\mathcal{R}_k$ and~$y\in\text{Supp}(f'_{k,b}(x_b))$, $\text{Inv}_{k}(t,b,y)=x_b$ and~$\text{Inv}_{k}(t,1-b,y)=x_{1-b}$.
                \item There exists an efficient deterministic algorithm~$\text{Chk}_{k}$ such that
                \begin{equation*}
                    \text{Chk}_{k}(b,x,y)=\mathbbm{1}_{y\in\text{Supp}(f'_{k,b}(x))}.
                \end{equation*}
                \item There exists some negligible function~$\eta$ such that
                \begin{equation*}
                    \underset{x\leftarrow\mathcal{X}}{\mathbbm{E}}\qty[\textbf{H}^{2}(f_{k,b}(x),f'_{k,b}(x))] \leq \eta (\lambda)
                \end{equation*}
                where~$\textbf{H}$ the Hellinger distance for distributions defined in Definition~\ref{def:hellinger}.
                \item There exists a quantum polynomial time (QPT) algorithm~$\text{Samp}_{k,b}$ that prepares the quantum state
                \begin{equation*}
                    \ket{\psi'}=\frac{1}{\sqrt{|\mathcal{X}|}} =
                    \sum_{x\in\mathcal{X},y\in\mathcal{Y}}
                    \sqrt{(f'_{k,b}(x))(y)}\ket{x}\ket{y}.
                \end{equation*}
            \end{enumerate}
            \item \textbf{Adaptive Hardcore Bit:} For all keys~$k\in\KF$, the following holds. For some integer~$w$ that is a polynomially bounded function of~$\lambda$
            \begin{enumerate}
                    \item For all~$b\in\qty{0,1}$ and~$x\in\mathcal{X}$, there exists a set~$G_{k,b,x}\subseteq\qty{0,1}^{w}$ such that~$\Pr_{d\leftarrow\qty{0,1}^w}[d\notin G_{k,b,x}]\leq\eta(\lambda)$ for some negligible function~$\eta$. In addition, there exists a PPT algorithm that checks for membership in~$G_{k,b,x}$ given~$k,b,x$ and the trapdoor~$t$.
                    \item Let
                    \begin{align*}
                        \begin{split}
                            H_k & \coloneqq \qty{(b,x_b,d,d\cdot(x_0\oplus x_1))|b\in\qty{0,1},(x_0,x_1)\in\mathcal{R}_{k},d\in G_{k,0,x_0} \cap G_{k,1,x_1}} \\
                            \bar{H}_k & \coloneqq \qty{(b,x_b,d,e)|e\in\qty{0,1}, (b,x_b,d,1-e)\in H_{k}}
                        \end{split}
                    \end{align*}
                    then for any QPT~$\mathcal{A}$ and polynomial size (potentially inefficient to prepare) advice state $\phi$ independent of the key $k$, there exists a negligible function~$\eta'$ for which the following holds\footnote{We remark that the same computational assumption is being used in~\cite{brakerski2021cryptographic}, even though the advice state is not included explicitly in the definition therein.}
                    \begin{equation}\label{eq:adaptive_hardcore_bit}
                        \qty|
                            \Pr_{(k,t)\leftarrow Gen_{\mathcal{F}(1^\lambda)}}
                            [\mathcal{A}(k, \phi)\in H_{k}]
                            -
                            \Pr_{(k,t)\leftarrow Gen_{\mathcal{F}(1^\lambda)}}
                            [\mathcal{A}(k, \phi)\in \bar{H}_{k}]
                            |
                        \leq \eta'(\lambda) \;.
                    \end{equation}
                \end{enumerate}
        \end{itemize}
    \end{definition}

    \begin{lemma}[Informal: existence of NTCF family{~\cite[Section 4]{brakerski2021cryptographic}}]\label{lem:ntcf_exist}
        There exists an LWE-based construction of an NTCF family as in Definition~\ref{def:ntcf}.\footnote{As previously noted, ~\cite{brakerski2021cryptographic} assumes that the LWE problem is hard also when given an advice state. Thus, though not explicitly stated, the construction in~\cite[Section 4]{brakerski2021cryptographic} is secure with respect to an advice state.}
    \end{lemma}   

    \section{Randomness certification}\label{sec:rand}
    Our main goal in this work is to provide a framework to lower bound the entropy accumulated during the execution of a protocol that uses a single quantum device.
The randomness generation protocol that we use is given as Protocol~\ref{protocol:simplified_accumulation}, where multiple rounds of interaction of a verifier with a quantum prover are performed.
The quantitative result of this section is a lower bound the amount of smooth min-entropy of the output of the protocol, namely, $H_{\min}^{\varepsilon}(\mathbf{\hat{\Pi}}|\mathbf{KTG}E)$.

The protocol and its security are based on the existence of an NTCF~$\mathcal{F}=\qty{f_{b,k}:\mathcal{X}\rightarrow\mathcal{Y}}$ (see Lemma~\ref{lem:ntcf_exist}).
We use the same assumption as in~\cite{brakerski2021cryptographic}, where the LWE problem~\cite{regev2009lattices} is exploited to construct an NTCF on which the protocol builds.
The definitions made in this chapter are stated implicitly with respect to~$\mathcal{F}$, a security parameter~$\lambda$ and a corresponding set of keys~$\KF$.

Before proving the validity of this protocol in the multiple round, we inspect a \emph{single} round of the protocol in Section~\ref{subsec:single_round_entropy}. 
This is done using a definition and inspection of a simplified one-round protocol and device, a reduction to single qubits and then the usage of the entropic uncertainty relation to define a min-tradeoff function. 
Following these steps, in Sections~\ref{subsec:entropy_acc_comp} and~\ref{sec:plots}, we use the results of Section~\ref{subsec:single_round_entropy} in combination with the EAT to prove a lower bound on the total amount of entropy accumulated during the execution of Protocol~\ref{protocol:simplified_accumulation}.

\begin{figure}
    \begin{protocol}{Entropy Accumulation Protocol\label{protocol:simplified_accumulation}} 
        \begin{enumerate}[label=\Alph*.]
            \item \textbf{Arguments:}
                \begin{itemize}
                    \item $D$ - untrusted device with which the verifier can interact according to the described protocol.
                    \item $\mathcal{F}$ - NTCF
                    \item $n\in\mathbb{N}$ - number of rounds
                    \item $\gamma\in(0,1]$ - expected fraction of test rounds
                    \item $\omega_{\text{exp}}$ - expected winning probability in an honest implementation
                    \item $\delta_{\text{est}}\in(0,1)$ - width of the confidence interval for parameter estimation
                    \item $\beta\in(0,1)$ - preimage test to equation test ratio
                \end{itemize}
            \item \textbf{Process:}
                \begin{enumerate}[label=\arabic*:]
                    \item Set all variables to the default value~$\perp$.
                    \item For every round~$i\in[n]$ do Steps 3-14.
                    \item \qquad Sample key and trapdoor~$(K_i, \tau_i)\leftarrow\text{Gen}_{\mathcal{F}}$ for the NTCF~$\mathcal{F}$.
                    \item \qquad Pass the key~$K_i$ to the device~$D$ which in return outputs~$Y_{i}\in\mathcal{Y}$.
                    \item \qquad Verifier samples~$G_i\leftarrow\text{Bernoulli}(1-\gamma)$.
                    \item \qquad If~$G_i = 0$: Test round
                    \item \qquad\qquad Verifier samples $~T_i\leftarrow\text{Bernoulli}(\beta)$ and passes it to the device~$D$.
                    \item \qquad\qquad If~$T_i = 0$,~$D$ outputs~$(b_i , x_i) \in \qty{0,1} \times\mathcal{X}$.
                    \item \qquad\qquad\qquad Verifier sets $\hat{\Pi}_i=(b_i - 2)\cdot\text{Chk}_{k_i}( b_i, x_i, Y_i) + 2$ and~$W_i = \mathbbm{1}_{\hat{\Pi}_i \neq 2}$.
                    \item \qquad\qquad If~$T_i = 1$,~$D$ outputs~$(u_i, d_i)\in\qty{0,1}\times\qty{0,1}^w$.
                    \item \qquad\qquad\qquad Verifier, using~$\tau_i$, sets~$\hat{M}_i=1$ if $(u_i, d_i)$ is a valid pair for the challenge (0 otherwise) and~$W_i=\hat{M}_i$.
                    \item \qquad If~$G_i = 1$: Generation round
                    \item \qquad\qquad Verifier sets~$T_i=0$ and passes~$T_i$ to the device~$D$.
                    \item \qquad\qquad\qquad $D$ outputs~$(b_i , x_i) \in \qty{0,1} \times\mathcal{X}$.
                    \item \qquad\qquad\qquad Verifier sets $\hat{\Pi}_i=(b_i - 2)\cdot\text{Chk}_{K_i}( b_i, x_i, Y_i) + 2$.
                    \item Abort if~$\sum_{j:G_j = 0} W_j < (\omega_{\text{exp}}\gamma-\delta_{\text{est}}) \cdot  n$.
                \end{enumerate}
            \item \textbf{Output:}
                $\hat{\Pi}_{j}$ for all $j$ such that $G_j=1$.
        \end{enumerate}
    \end{protocol}    
\end{figure}

\subsection{Single-round entropy}\label{subsec:single_round_entropy}

    In this Section we analyze a single round of Protocol~\ref{protocol:simplified_accumulation}, presented as Protocol~\ref{protocol:single_round}.

    \subsubsection{One round protocols and devices}

       Protocol~\ref{protocol:single_round}, roughly, describes a single round of Protocol~\ref{protocol:simplified_accumulation}.
        The goal of Protocol~\ref{protocol:single_round} can be clarified by thinking about an interaction between a verifier~$\mathcal{V}$ and a quantum prover~$\mathcal{P}$ in an independent and identically distributed (IID) scenario, in which the device is repeating the same actions in each round.
        Upon multiple rounds of interaction between the two,~$\mathcal{V}$ is convinced that the winning rate provided by~$\mathcal{P}$ is as high as expected and therefore can continue with the protocol. We will, of course, use the protocol later on without assuming that the device behaves in an IID manner throughout the multiple rounds of interaction.

      \begin{figure}
        \begin{protocol}{Single-Round Protocol \label{protocol:single_round}}
            \begin{enumerate}[label=\Alph*.]
                \item \textbf{Arguments:}
                    \begin{itemize}
                        \item Classical verifier~$\mathcal{V}$
                        \item Quantum prover~$\mathcal{P}$ represented by the device~$D=\qty(\phi,\Pi,M)$
                        \item $\mathcal{F}$ - NTCF
                        \item $\beta\in(0,1)$
                    \end{itemize}
                \item \textbf{Process:}
                    \begin{enumerate}[label=\arabic*:]
                        \item $\mathcal{V}$ samples $(k, \tau)\leftarrow\text{Gen}_{\mathcal{F}}$ for the NTCF~$\mathcal{F}$.
                        \item $\mathcal{V}$ sends the key~$k$ to the prover.
                        \item $\mathcal{P}$ sends~$y\in\mathcal{Y}$ to~$\mathcal{V}$.
                        % \item $\mathcal{V}$ uses the trapdoor~$\tau$ to compute the preimages,~${x}_b$, of the image~$y$.
                        \item $\mathcal{V}$ samples~$T\leftarrow\text{Bernoulli}( \beta )$ and gives it to~$\mathcal{P}$.
                        \item If~$T=0$
                        \item \qquad $\mathcal{P}$ returns~$(b,x)\in\qty{0,1}\times\mathcal{X}$ to~$\mathcal{V}$.
                        \item \qquad $\mathcal{V}$ sets~$\hat{\pi}=(b - 2)\cdot\text{Chk}_{k}( b, x, y) + 2$ and~$W=\text{Chk}_{k}(b,x,y)$.
                        \item If~$T=1$
                        \item \qquad $\mathcal{P}$ returns~$(u,d)\in\qty{0,1}\times\qty{0,1}^{w}$ to~$\mathcal{V}$.
                        \item \qquad $\mathcal{V}$, using~$\tau$, sets~$W=1$ if~$(u,d)$ is a valid answer to the challenge.
                    \end{enumerate}
            \end{enumerate}
        \end{protocol}
    \end{figure}

      We begin by formally defining the most general device that can be used to execute the considered single round Protocol~\ref{protocol:single_round}. 
    \begin{definition}[General device~\cite{brakerski2021cryptographic}]\label{def:general_device}
        A general device is a tuple~$D=(\phi,\Pi,M)$ that receives~$k\in\KF$ as input and specified by the following:
        \begin{enumerate}
            \item A normalized density matrix~$\phi\in\mathcal{H}_{D}\otimes\mathcal{H}_{Y}$.
            \begin{itemize}
                \item $\mathcal{H}_{D}$ is a polynomial (in $\lambda$) space, private to the device.
                \item $\mathcal{H}_{Y}$ is a space private to the device whose dimension is the same of the cardinality of~$\mathcal{Y}$.
                \item For every~$y\in\mathcal{Y}$,~$\phi_{y}$ is a sub-normalized state such that
                    \begin{equation*}
                        \phi_{y} = (\mathbbm{1}_{D}\otimes\bra{y}_{Y})\phi(\mathbbm{1}_{D}\otimes\ket{y}_{Y}).
                    \end{equation*}
            \end{itemize}
            \item For every~$y\in\mathcal{Y}$, a projective measurement~${M_{y}^{(u,d)}}$ on~$\mathcal{H}_{D}$, with outcomes~$(u,d)\in\qty{0,1}\times\qty{0,1}^w$.
            \item For every~$y\in\mathcal{Y}$, a projective measurement~${\Pi_{y}^{(b,x)}}$ on~$\mathcal{H}_{D}$, with outcomes~$(b,x)\in\qty{0,1}\times\mathcal{X}$.
            For each~$y$, this measurement has two designated outcomes~$(0,x_0),(1,x_1)$.
        \end{enumerate}
    \end{definition}

    \begin{definition}[Efficient device]\label{def:efficient_device}
        We say that a device $D=(\phi,\Pi,M)$ is efficient if:
        \begin{enumerate}
            \item The state $\phi$ is a polynomial size (in $\lambda$) ``advice state'' that is independent of the chosen keys $k\in\KF$. The state might not be producible using a polynomial time quantum circuit (we say that the state can be inefficient).
            \item The measurements $\Pi$ and $M$ can be implemented by polynomial size quantum circuit.
        \end{enumerate}
    \end{definition}

    We emphasize that the device~$D$ is computationally bounded by both time steps and memory.
    This prevents pre-processing schemes in which the device manually goes over all keys and pre-images to store a table of answers to all the possible challenges as such schemes demand exponential memory in~$\lambda$.

    The cryptographic assumption made on the device is the following.    
    Intuitively, the lemma states that due to the hardcore bit property in Equation~\eqref{eq:adaptive_hardcore_bit}, the device cannot pass both pre-image and equation tests; once it passes the pre-image test, trying to pass also the equation test results in two computationally indistinguishable state-- one in which the device also passes the equation test and one in which it does not.
    
    \begin{lemma}[Computational indistinguishability{~\cite[Lemma~7.1]{brakerski2021cryptographic}}]\label{lemma:computational_indistinguishability}
        Let~$D=(\phi, \Pi, M)$ be an efficient general device, as in Definitions~\ref{def:general_device} and~\ref{def:efficient_device}. Define a sub-normalized density matrix
        \begin{equation*}
            \tilde{\phi}_{YBXD} =
                \sum_{y\in\mathcal{Y}} \dyad{y}_{Y} \otimes
                \sum_{b\in\qty{0,1}} \dyad{b,x_b}_{BX} \otimes
                \Pi_y ^{(b, x_b)} \phi_y \Pi_y ^{(b, x_b)} \; .
        \end{equation*}
        Let
        \begin{align*}
            \begin{split}
                \sigma_0 & = 
                    \sum_{b\in\qty{0,1}} \dyad{b,x_b}_{BX} \otimes
                    \sum_{(u,d)\in V_{y,1}} \dyad{u,d}_U \otimes
                    (\mathbbm{1}_Y \otimes M_y^{(u,d)})
                    \tilde{\phi}^{(b)}_{YD}
                    (\mathbbm{1}_Y \otimes M_y^{(u,d)}) \; ,
                \\
                \sigma_1 & = 
                    \sum_{b\in\qty{0,1}} \dyad{b,x_b}_{BX} \otimes
                    \sum_{(u,d)\notin V_{y,1}} \mathbbm{1}_{d\in\hat{G}_y}\dyad{u,d}_U \otimes
                    (\mathbbm{1}_Y \otimes M_y^{(u,d)})
                    \tilde{\phi}^{(b)}_{YD}
                    (\mathbbm{1}_Y \otimes M_y^{(u,d)}) \; ,
            \end{split}
        \end{align*}
        where~$V_{y,1}$ is the set of valid answers to challenge~$1$ (equation test).
        Then,~$\sigma_0$ and~$\sigma_1$ are computationally indistinguishable.
    \end{lemma}
        
        The proof is given in~\cite[Lemma~7.1]{brakerski2021cryptographic}.\footnote{Note that even though the proof in~\cite{brakerski2021cryptographic} does not address the potentially inefficient advice state $\phi$ explicitly, it holds due to the non-explicit definition of their computational assumption.}

    \begin{figure}
        \begin{protocol}{Simplified Single-Round Protocol \label{protocol:simplified_single_round}} 
            \begin{enumerate}[label=\Alph*.]
                \item \textbf{Arguments:}
                    \begin{itemize}
                        \item Classical verifier~$\mathcal{V}$
                        \item Quantum prover~$\mathcal{P}$ represented by the simplified device~$\tilde{D}=\qty(\phi,\tilde{\Pi},\tilde{M})$
                        \item $\beta\in(0,1)$
                    \end{itemize}
                \item \textbf{Process:}
                    \begin{enumerate}[label=\arabic*:]
                        \item $\mathcal{V}$ samples the random variable $T \leftarrow \text{Bernoulli} \qty( \beta )$ which is then passed to~$\mathcal{P}$.
                        \item If $T=0$
                        \item \qquad $\mathcal{P}$ measures~$\tilde{\Pi}$ on the state~$\phi$ for an outcome~$p\in\qty{0,1,2}$.
                        \item If $T=1$
                        \item \qquad $\mathcal{P}$ measures $\tilde{M}$ on the state~$\phi$ for an outcome~$m\in\qty{0,1}$.
                    \end{enumerate}
            \end{enumerate}
        \end{protocol}
    \end{figure}

    We proceed with a reduction of Protocol~\ref{protocol:single_round} to a simplified one, Protocol~\ref{protocol:simplified_single_round}.
    As the name suggests, it will be easier to work with the simplified protocol and devices when bounding the produced entropy.
    We remark that a similar reduction is used in~\cite{brakerski2021cryptographic}; the main difference is that we are using the reduction on the level of a single round, in contrast to the way it is used in~\cite[Section~8]{brakerski2021cryptographic} when dealing with the full multi-round protocol. Using the reduction in the single round protocol instead of the full protocol helps disentangling the various challenges that arise in the analysis of the entropy.

    \begin{definition}[Simplified device{~\cite[Definition 6.4]{brakerski2021cryptographic}}]\label{def:simplified_device}
        A simplified device is a tuple~$\tilde{D}=\qty(\phi,\tilde{\Pi},\tilde{M})$ that receives~$k\in\KF$ as input and specified by the following:
        \begin{enumerate}
            \item $\phi={\{\phi_{y}\}}_{y\in\mathcal{Y}}\subseteq\text{Pos}\qty(\mathcal{H}_D)$ is a family of positive
            semidefinite operators on an arbitrary space~$\mathcal{H}_D$ such that~$\sum_{y}\text{Tr}(\tilde{\phi}_y)\leq 1$;
            \item $\tilde{\Pi}$ and~$\tilde{M}$ are defined as the sets~$\qty{\tilde{\Pi}_{y}}_{y},\qty{\tilde{M}_{y}}_{y}$ respectively such that for each~$y\in\mathcal{Y}$, the operators,~$\tilde{M}_{y}=\qty{\tilde{M}_y^0,\tilde{M}_y^1=\mathbbm{1}-\tilde{M}_y^0}$ and $\tilde{\Pi}_{y} = \qty{\tilde{\Pi}_y^0,\tilde{\Pi}_y^1,\tilde{\Pi}_y^2=\mathbbm{1}-\tilde{\Pi}_y^0 - \tilde{\Pi}_y^1}$, are projective measurements on~$\mathcal{H}_D$.
        \end{enumerate}
    \end{definition}
    
    \begin{definition}[Simplified device construction~\cite{brakerski2021cryptographic}]\label{def:simplified_device_construction}
        Given a general device as in Definition~\ref{def:general_device} $D=(\phi,\Pi,M)$, we construct a simplified device~$\tilde{D}=(\phi,\tilde{\Pi},\tilde{M})$ in the following manner:
        \begin{itemize}
            \item The device~$\tilde{D}$ measures~$y\in\mathcal{Y}$ like the general device~$D$ would.
            \item The measurement~$\tilde{\Pi}=\{\tilde{\Pi}_{y}^{0},\tilde{\Pi}_{y}^{1},\tilde{\Pi}_{y}^{2}\}$ is defined as follows.
            \begin{itemize}
                \item Perform the measurement~$\{\Pi_{y}^{(b,x)}\}_{b\in\qty{0,1},x\in\mathcal{X}}$ for an outcome~$(b,x)$.
                \item If~$\mathrm{Chk}_k (b,x,y)=1$, the constructed device returns~$b$ corresponding to the projection~$\tilde{\Pi}_{y}^{b}\in\qty{\tilde{\Pi}_{y}^{0}, \tilde{\Pi}_{y}^{1}}$.
                \item If~$\mathrm{Chk}_k (b,x,y)=0$, the constructed device returns 2 corresponding to the projection~$\tilde{\Pi}_{y}^{2}$.
            \end{itemize}
            \item The measurement~$\tilde{M}=\{\tilde{M}^0_y,\tilde{M}^1_y\}$ is defined as follows.
            \begin{align*}
                \tilde{M}^0_y = \sum_{(u,d)\in V_{y,1}} M_{y}^{(u,d)}
                ~~~,~~~&
                \tilde{M}^1_y = \mathbbm{1} - \tilde{M}^0_y
                \; ,
            \end{align*}
            where~$V_{y,1}$ is valid answers for the equation test.
            Meaning the outcome~$\tilde{M}=0$ corresponds to a valid response~$(u,d)$ in the equation test.
        \end{itemize}
    \end{definition}

    The above construction of a simplified device maintains important properties of the general device. 
    Firstly, the simplified device fulfils the same cryptographic assumption as the general one. This is stated in the following corollary.
    
    \begin{corollary}\label{cor:comp_assump_simplified}
        Given a general efficient device $D=(\phi,\Pi,M)$, a simplified device $\tilde{D}=(\phi,\tilde{\Pi},\tilde{M})$ constructed according to Definition~\ref{def:simplified_device_construction} is also efficient. Hence, the cryptographic assumption described in Lemma~\ref{lemma:computational_indistinguishability} holds also for $\tilde{D}$ as well. 
    \end{corollary}

    Secondly, the entropy produced by the simplified device in the simplified single-round protocol, Protocol~\ref{protocol:simplified_single_round}, is identical to that produced by the general device in single round protocol, Protocol~\ref{protocol:single_round}.
    A general device executing Protocol~\ref{protocol:single_round} defines a probability distribution of~$\hat{\pi}$ over~$\qty{0,1,2}$.
    Using the same general device to construct a simplified one, via Definition~\ref{def:simplified_device_construction}, leads to the same distribution for~$\tilde{\Pi}$ when executing Protocol~\ref{protocol:simplified_single_round}.
    This results in the following corollary.
    
    \begin{corollary}\label{cor:entropy_reduction}
        Given a general efficient device $D=(\phi,\Pi,M)$ and a simplified device $\tilde{D}=(\phi,\tilde{\Pi},\tilde{M})$ constructed according to Definition~\ref{def:simplified_device_construction}, we have for all $k$,
        \begin{equation*}
            H(\hat{\pi}|EY)^{\text{General}} =   H(\tilde{\Pi}|EY)^{\text{Simplified}} \;,
        \end{equation*}
        where $H(\hat{\pi}|EY)^{\text{General}}$ is the entropy produced by Protocol~\ref{protocol:single_round} using the general device $D$ and $H(\tilde{\Pi}|EY)^{\text{Simplified}}$ is the entropy produced in Protocol~\ref{protocol:simplified_single_round} using the simplified device $\tilde{D}$.
        Both entropies are evaluated on the purification of the state $\phi$.
    \end{corollary}
    The above corollary tells us that we can reduce the analysis of the entropy created by the general device to that of the simplified one, hence justifying its construction and the following sections.

    \subsubsection{Reduction to qubits}\label{sec:reduction_to_qubits}
    In order to provide a clear understanding of the quantum uncertainty that arises from two measurements, it is beneficial to examine the square overlap between those measurements (Definition~\ref{def:square_overlap}).
    The Bloch sphere representation, which pertains to Hilbert spaces of two dimensions, offers a lucid illustration of this concept.
    
    Throughout this section, it is demonstrated that the devices under investigation, under specific conditions, can be expressed as a convex combination of devices, each operating on a single qubit.
    Working in a qubit subspace then allows one to make definitive statements regarding the entropy of the measurement outcomes produced by the device.
    We remark that this is in complete analogy with the proof techniques used when studying DI protocols in the non-local setting, in which one reduces the analysis to that of two single qubit devices~\cite[Lemma~1]{Pironio_2009}.
    
        \begin{lemma}\label{lemma:simplified_qubit_decomposition}
            Let~$\tilde{D}=(\phi,\tilde{\Pi},\tilde{M})$ be a simplified device (Definition~\ref{def:simplified_device}), acting within a Hilbert space~$\mathcal{H}$ of a countable dimension, with the additional assumption that~$\tilde{\Pi}$ consists of only 2 outcomes and let~$\Gamma$ be a Hermitian projection that commutes with both~$\tilde{\Pi}$ and~$\tilde{M}$.
            Given an operator~$F=f(\tilde{M},\Gamma)$ constructed from some non-commutative polynomial of 2 variables~$f$, let~$S_{\tilde{D}} = \langle{f(\tilde{M},\Gamma)}_\phi \rangle$ be the expectation value of~$F$.
            % \begin{equation*}
            %     S_D = \expval{f(M,K)}_\phi \; .
            % \end{equation*}
            Then, there exists a set of Hermitian projections~$\qty{B_{j}}_j$, acting within the same Hilbert space~$\mathcal{H}$, satisfying the following conditions
            \begin{equation}\label{eq:projector_conditions}
                \begin{array}{lllll}
                    1)~  \forall j, ~ \operatorname{Rank}\qty(B_{j}) \leq 2 \\ \\
                    2)~  \sum_{j}B_{j}=\mathbbm{1} \\ \\
                    3)~  \forall j, ~ \qty[\tilde{\Pi},B_{j}] = \qty[\tilde{M}, B_{j}] = \qty[\Gamma,B_{j}] = 0                
                \end{array}
            \end{equation}
            such that 
            \begin{equation*}
                S_{\tilde{D}} = \sum_{j} \Pr\qty(j) S_{{\tilde{D}}_{j}} \; ,
            \end{equation*}
            where~$\Pr(j) = \Tr\qty(B_{j} \phi)$ and~$S_{{\tilde{D}}_{j}}$ is the expectation value of~$F$ given the state~$\frac{B_{j}\phi B_{j}}{\Tr\qty[B_{j}\phi B_{j}]}$, corresponding to the simplified device~${\tilde{D}}_{j} = \qty(\frac{B_{j}\phi B_{j}}{\Tr\qty[B_{j}\phi B_{j}]},\tilde{\Pi},\tilde{M})$.
            That is, $ S_{{\tilde{D}}_j}=\langle f(\tilde{M},\Gamma)\rangle_{B_{j}\phi B_{j}/\Tr[B_{j}\phi B_{j}]}$.
        \end{lemma}
        \begin{proof}
            As an immediate result of Lemma~\ref{lemma:jordans_projections}, there exists a basis in which~$\tilde{\Pi},\tilde{M}$ and~$\Gamma$ are~$2\times 2$ block diagonal.
            In this basis, we take the projection on every block~$j$ as~$B_j$.
            This satisfies the conditions in Conditions~\eqref{eq:projector_conditions}.
            Furthermore,
            \[
                S_{\tilde{D}} =  \tr\qty(\sum_j F B_j \phi B_j)  =  \sum_j \Pr(j) \tr\qty( F  \frac{B_j \phi B_j}{\tr(B_j \phi B_j)}) =  \sum_j \Pr(j) S_{{\tilde{D}}_j}\;.
            \] \qedhere
        \end{proof}

        Note that the simplified device~$\tilde{D}_{j}=\qty(\frac{B_{j}\phi B_{j}}{\Tr\qty[B_{j}\phi B_{j}]},\tilde{\Pi},\tilde{M})$ yields the same expectation values to those of the simplified device~$\tilde{D}_{j}=\qty(\frac{B_{j}\phi B_{j}}{\Tr\qty[B_{j}\phi B_{j}]},B_{j} \tilde{\Pi} B_{j},B_{j}\tilde{M}B_{j})$.
        The resulting operation is therefore, effectively, done in a space of a single qubit. In addition, due to the symmetry of~$\tilde{\Pi}$ and~$\tilde{M}$ in this proof, the lemma also holds for an observable constructed from~$\tilde{\Pi}$ and~$\Gamma$ instead of~$\tilde{M}$ and~$\Gamma$. I.e.~$S_{\tilde{D}} = \expval{f(\tilde{\Pi},\Gamma)}_{\phi}$.

    The simplified protocol permits the use of the uncertainty principle, appearing in Lemma~\ref{lemma:uncertainty_relations}, in a vivid way since it has a geometrical interpretation on the Bloch sphere; recall Figure~\ref{fig:uncertainty_bloch}.
    Under the assumption that~$\tilde{\Pi}$ has two outcomes (this has yet to be justified)  we can represent~$\tilde{\Pi}$ and~$\tilde{M}$ as two Bloch vectors with some angle between them that corresponds to their square overlap -- The smaller the square overlap, the closer the angle is to~$\pi/2$.
    In the ideal case, the square overlap is~$1/2$ which means that in some basis the two measurements are the standard and the Hadamard measurements.
    If one is able to confirm that~$\mathrm{Pr}(\tilde{M}=0)=1$, the only possible distribution on the outcomes of~$\tilde{\Pi}$ is a uniform one, which has the maximal entropy.
    
    $\tilde{\Pi}$, however, has 3 outcomes and not 2, rendering the reduction to qubits unjustified.
    We can nonetheless argue that the state being used in the protocol is very close to some other state which produces only the first 2 outcomes.
    The entropy of both states can then be related using Equation~\eqref{eq:entropy_continuity}.
    The ideas described here, are combined and explained thoroughly in the main proof shown in the following subsection.

\subsubsection{Conditional entropy bound}

    The main proof of this subsection is done with respect to a simplified device constructed from a standard one using Definition~\ref{def:simplified_device_construction}.
    
    Before proceeding with the proof, we define the winning probability in both challenges.

    \begin{definition}\label{def:winning_rates}
        Given a simplified device (with implicit key~$k$) $\tilde{D}=\qty(\phi, \tilde{\Pi}, \tilde{M})$, for a given~$y\in\mathcal{Y}$ we define the~${\tilde{\Pi}}_y$ and~${\tilde{M}}_y$ winning probabilities, respectively, as
        \begin{equation*}
             \omega_p^y \coloneqq \frac{\Tr((\mathbbm{1} - {\tilde{\Pi}}_{y}^{2})\phi_y)}{\Tr(\phi_y)} \qquad \;; \qquad 
             \omega_m^y \coloneqq \frac{\Tr({\tilde{M}}_{y}^{0}\phi_y)}{\Tr(\phi_y)} \; .
        \end{equation*}
        Likewise, the winning probabilities of~$\Pi$ and~$M$ as 
        \begin{equation*}
             \omega_p \coloneqq \sum_{y} \Tr((\mathbbm{1} - {\tilde{\Pi}}_{y}^{2})\phi_y) \qquad \;; \qquad 
             \omega_m \coloneqq \sum_{y} \Tr({\tilde{M}}_{y}^{0}\phi_y) \; .
        \end{equation*}
        Recall that~$\phi_y$ are sub-normalized.
    \end{definition}

    We are now ready to prove our main technical lemma. 

    \begin{lemma}\label{lemma:entropy_bound_proof}
        Let~$\tilde{D}=(\phi, \tilde{\Pi}, \tilde{M})$ be a simplified device as in Definition~\ref{def:simplified_device}, constructed from an efficient general device~$D=\left(\phi, \Pi, M\right)$ in the manner depicted in Definition~\ref{def:simplified_device_construction}. Let~$\Phi_y\in\mathcal{H}_{D}\otimes\mathcal{H}_{E}$ be a purification of~${\phi_y}$ (respecting the sub-normalization of~$\phi_y$) and $\Phi=\{\Phi_y\}_{y\in\mathcal{Y}}$. 
        For all~$c\in(1/2,1]$ and some negligible function~$\xi(\lambda)$, the following inequality holds:
        \begin{align}\label{eq:simplified_entropy_lower_bound}
            H(\tilde{\Pi}|E,Y)_{\Phi} \geq
                & \max\qty{0, 1-\sqrt{2}A(c)\sqrt[4]{1-\omega_p}-A(c)\sqrt{1-\omega_m}} \\
                & \quad \times \qty(\log_{2}\qty(1/c) - h\qty(\omega_m - 2\sqrt{1-\omega_p}-\sqrt{2}A(c)\sqrt[4]{1-\omega_p}-A(c)\sqrt{1-\omega_m})) \nonumber \\
                & \quad - \sqrt{1-\omega_p} \log_2 (3) - \qty(1+\sqrt{1-\omega_p}) h\qty(\frac{\sqrt{1-\omega_p}}{1+\sqrt{1-\omega_p}}) - \xi(\lambda) \nonumber \;,
        \end{align}
        where~$h(\cdot)$ is the binary entropy function and
        \begin{equation}\label{eq:ac}
            A(c)\coloneqq10/{(2c-1)}^2 \; .
        \end{equation}
    \end{lemma}    

    \begin{figure}
        \centering
        \begin{subfigure}[b]{0.45\textwidth}
            \includegraphics[scale=0.65]{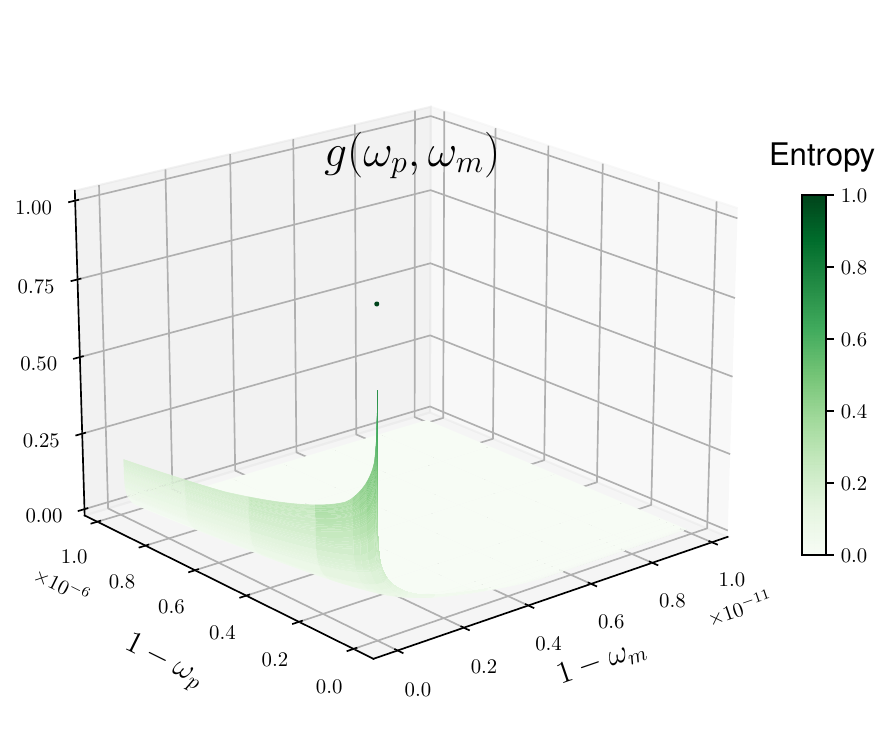}
            \caption
            {
                Plot of the function given in Equation~\eqref{eq:two_variable_bound}.
            }
            \label{sub_fig:von_neumann_2}
        \end{subfigure}
        \hfill
        \begin{subfigure}[b]{0.45\textwidth}
            \includegraphics[scale=0.55]{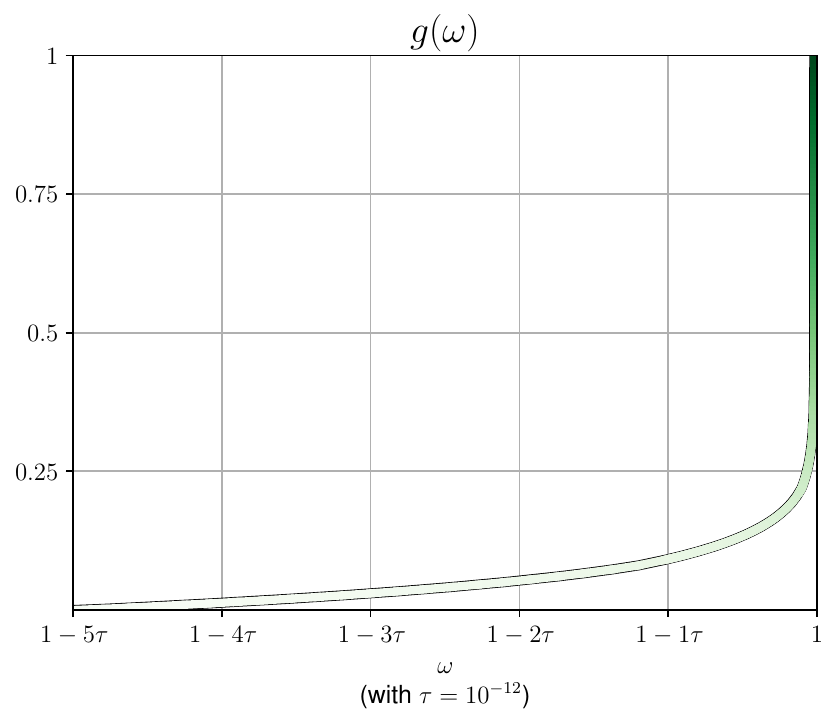}
            \caption
            {
                Plot of the function appearing in Equation~\eqref{eq:single_variable_bound} with~$\beta = 0.045$.
                The width of the curve is added for clarity.
            }
            \label{sub_fig:von_neumann}
        \end{subfigure}
        \caption
        {
            The entropy as a function of the winning probabilities, as given by Equations~\eqref{eq:two_variable_bound} and~\eqref{eq:single_variable_bound}. In the left plot, the differential at~$(\omega_p = 1, \omega_m = 1)$ diverges and therefore one does not see the entropy grows up to its optimal point (the dark green point); this is only a visual effect.
            This entropy and the divergence is seen more clearly on a slice of this graph, appearing in right plot.
        }
        \label{fig:pi_von_neumann_both}
    \end{figure}
    
    \begin{proof}
        For every $y\in\mathcal{Y}$ we introduce the state
        \begin{equation*}
            {\Psi}_{y} \coloneqq
                \frac{(\tilde{\Pi}_{y}^{0}+\tilde{\Pi}_{y}^{1}){\Phi}_{y}(\tilde{\Pi}_{y}^{0}+\tilde{\Pi}_{y}^{1})}
                {\Tr((\tilde{\Pi}_{y}^{0}+\tilde{\Pi}_{y}^{1}){\Phi}_{y})}
        \end{equation*}
        in order to reduce the problem to a convex combination of 2-dimensional ones, as described in Section~\ref{sec:reduction_to_qubits}. We provide a lower bound for the entropy~$H(\tilde{\Pi}|E,Y)_{{\Psi_y}}$ and by continuity of entropies, a lower bound for~$H(\tilde{\Pi}|E,Y)_{{\Phi_y}}$ is then derived.
        The proof proceeds in steps.
        \begin{enumerate}
            \item Let
                \begin{equation*}
                    U_{y}^{0} = \tilde{\Pi}_{y}^{0}-(\tilde{\Pi}_{y}^{1}+\tilde{\Pi}_{y}^{2}) \qquad  ;\qquad U_{y}^{1}=\tilde{M}_{y}^{0}-\tilde{M}_{y}^{1} \;.
                \end{equation*}
                Using Jordan's lemma, Lemma~\ref{lemma:jordans_extension}, there exists an orthonormal basis where the operators are $2\times 2$-block diagonal. Let~$\Gamma_y$ be the Hermitian projection on blocks where the square overlap of~$U_{y}^{0}$ and~$U_{y}^{1}$ is bound by~$c$ (good blocks).
                Note that~$\Gamma_y$ also commutes with both unitaries.
            \item Using Lemma~\ref{lemma:good_subspace_projection} we bound the probability to be in subspace where the square overlap of~$\Pi$ and~$M$ is larger than~$c$ (bad blocks).
                \begin{equation*}
                    \text{Tr}((\mathbbm{1} - \Gamma_y){\Psi}_{y}) \leq \frac{2\mu+10\sqrt{\text{Tr}(\tilde{M}^{1}_{y}{\Psi}_{y})}}{(2c-1)^2}
                    = A(c)\qty(\frac{1}{5}\mu+\sqrt{\text{Tr}(\tilde{M}^{1}_{y}{\Psi}_{y})})
                \end{equation*}
                where~$A(c)$ is given by Equation~\eqref{eq:ac} and
                \begin{align*}
                    \mu & \coloneqq \qty|\frac{1}{2} - \Tr(\tilde{M}_{y}^{0} \tilde{\Pi}_{y}^{0}{\Psi}_{y}\tilde{\Pi}_{y}^{0})-\Tr(\tilde{M}_{y}^{0} \tilde{\Pi}_{y}^{1}{\Psi}_{y}\tilde{\Pi}_{y}^{1})| \\
                    & = \frac{1}{2}
                    \qty
                    |
                        \sum_{b\in\qty{0,1}} \tr\qty(\tilde{\Pi}_{y}^{b} \Psi_{y} \tilde{\Pi}_{y}^{b}) - 
                        2\sum_{b\in\qty{0,1}} \tr\qty(\tilde{M}_{y}^{0}\tilde{\Pi}_{y}^{b} \Psi_{y} \tilde{\Pi}_{y}^{b})
                    | \\
                    & = \frac{1}{2}
                    \qty
                    |
                        \sum_{b\in\qty{0,1}} \tr\qty(\tilde{M}_{y}^{0}\tilde{\Pi}_{y}^{b} \Psi_{y} \tilde{\Pi}_{y}^{b}) - 
                        \sum_{b\in\qty{0,1}} \tr\qty((\mathbbm{1}-\tilde{M}_{y}^{0})\tilde{\Pi}_{y}^{b} \Psi_{y} \tilde{\Pi}_{y}^{b})
                    | \; .
                \end{align*}
                Due to Lemma~\ref{lemma:computational_indistinguishability} and Corollary~\ref{cor:comp_assump_simplified},~$\mu$ is negligible in the security parameter~$\lambda$.
                Thus, for some negligible function~$\eta$,
                \begin{equation}\label{eq:k_projection_original}
                    \text{Tr}((\mathbbm{1} - \Gamma_y){\Psi}_{y}) \leq A(c)\sqrt{\text{Tr}(\tilde{M}^{1}_{y}{\Psi}_{y})}+\eta(\lambda) \; .
                \end{equation}
            \item Note that $\Phi$ is the state of the device (and the adversary), not $\Psi$. Hence,  we cannot, a priori, relate~$\text{Tr}(\tilde{M}^{1}_{y}{\Psi}_{y})$ to the winning probabilities of the device.
            We therefore want to translate Equation~\eqref{eq:k_projection_original} to quantities observed in the application of Protocol~\ref{protocol:simplified_single_round} when using the simplified device~$\tilde{D}$.
            That is, we would like to use the values~$\omega_p, \omega_m$ given in Definition~\ref{def:winning_rates} in our equations:
                \begin{align*}
                    \sqrt{\text{Tr}(\tilde{M}^{1}_{y}{\Psi}_{y})} & = \sqrt{\text{Tr}(\tilde{M}^{1}_{y}({\Psi}_{y} - {\Phi}_y))+\text{Tr}(\tilde{M}^{1}_{y}{\Phi}_y)} \\ 
                    % & \leq
                    %     \sqrt{\text{Tr}(\tilde{M}^{1}_{y}|\tilde{\Psi}_{y} - \tilde{\Phi}_y|)} +
                    %     \sqrt{\text{Tr}(\tilde{M}^{1}_{y}\tilde{\Phi}_y)}
                    %     \\
                    & \leq
                        \sqrt{\text{Tr}(|{\Psi}_{y} - {\Phi}_y|)} +
                        \sqrt{\text{Tr}(\tilde{M}^{1}_{y}{\Phi}_y)}
                        \\
                    & \leq
                            \sqrt{2\sqrt{\text{Tr}(\tilde{\Pi}_y^2 {\Phi}_y)}} +
                            \sqrt{\text{Tr}(\tilde{M}^{1}_{y} {\Phi}_y)} 
                            \\
                    & = 
                        \sqrt{2\sqrt{1-\omega_p^y}} +
                        \sqrt{1-\omega_m^y}
                        \;,
                \end{align*}
            where in the last inequality we used the gentle measurement lemma~\cite[Lemma 9.4.1]{wilde2013gentle}.
            The last equation, combined with Equation~\eqref{eq:k_projection_original}, immediately yields
            \begin{equation}\label{eq:good_subspace_bound}
                \text{Tr}(\Gamma_y {\Psi}_{y}) \geq 1 - A(c)\qty(\sqrt{2}\sqrt[4]{1-\omega_p^y} +\sqrt{1-\omega_m^y}) - \eta(\lambda) \; .
            \end{equation}
            \item In a similar manner, for later use, we must find a lower bound on~$\text{Tr}(\tilde{M}_y^0 {\Psi}_{y})$ using quantities that can be observed from the simplified protocol. To that end, we again use the gentle measurement lemma:
                \begin{align*}
                    \text{Tr}(\tilde{M}_y^0 {\Phi}_y) & = \text{Tr}(\tilde{M}_y^0( {\Phi}_{y} - {\Psi}_{y})) + \text{Tr}(\tilde{M}_y^0 {\Psi}_{y}) \\
                    & \leq 2\sqrt{(\mathbbm{1} - {\tilde{\Pi}}_y^2) {\Phi}_y} + \text{Tr}(\tilde{M}_y^0 {\Psi}_{y}) \; .
                \end{align*}
                Using the definitions of~$\omega_m^y,\omega_p^y$,
                \begin{flalign}
                    & \phantom{~\Rightarrow} \omega_m^y  \leq 2\sqrt{1-\omega_p^y} + \text{Tr}(\tilde{M}_y^0 {\Psi}_{y}) \nonumber \\
                    & \label{eq:mw_bound}\Rightarrow \text{Tr}(\tilde{M}_y^0 {\Psi}_{y}) \geq \omega_m^y - 2\sqrt{1-\omega_p^y}
                \end{flalign}
            \item We proceed by providing a bound on the conditional entropy, given that~$\Gamma_y =0$, using the uncertainty principle in Lemma~\ref{lemma:uncertainty_relations}.
            Conditioned on being in a good subspace (which happens with probability~$\Pr(\Gamma_y = 0)$), the square overlap of~$\tilde{\Pi}$ and~$\tilde{M}$ is upper bounded by~$c$. We can then bound the entropy of~$\tilde{\Pi}$ using the entropy of~$\tilde{M}$:
            \begin{align*}
                H(\tilde{\Pi}|E,Y=y,\Gamma=0)_{{\Psi}_{y}} &
                \geq \log_2 (1/c) - H(\tilde{M}|Y=y,\Gamma=0)_{{\Psi}_{y}} \\
                & = \log_2 (1/c) - h\qty(\text{Pr}(\tilde{M}_{y} = 0| \Gamma_y = 0)_{\Psi_y}).
            \end{align*}
            Since we have a lower bound on~$\text{Tr}(\tilde{M}_y^0 {\Psi}_{y})$, we proceed by working in the regime where the binary entropy function is strictly 
            decreasing. To that end, it is henceforth assumed that the argument of the binary entropy function, and all of its subsequent lower bounds, are larger than~$1/2$.  
            By using the inequality
            \begin{equation}\label{eq:barbaric_ineq}
                \text{Pr}(\tilde{M}_{y} = 0| \Gamma_y = 0) \geq \text{Pr}(\tilde{M}_{y} = 0) + \text{Pr}(\Gamma_y = 0) - 1
            \end{equation}
            we obtain
            \begin{equation*}
                - h\qty(\text{Pr}(\tilde{M}_{y} = 0| \Gamma_y = 0)) \geq 
                - h\qty(\text{Pr}(\tilde{M}_{y} = 0) + \text{Pr}(\Gamma_y = 0) - 1) \; .
            \end{equation*}
            Therefore, 
            \begin{equation}\label{eq:good_subspace_entropy}
                H(\tilde{\Pi}|E,Y=y,\Gamma=0)_{{\Psi}_{y}} \geq \log_2 (1/c) - h\qty(\text{Pr}(\tilde{M}_{y} = 0) + \text{Pr}(\Gamma_y = 0) - 1) \; .
            \end{equation}
            \item We now want to bound the value~$H(\tilde{\Pi}| E, Y)_{{\Psi}_{y}}$, i.e., without conditioning on the event~$\Gamma_y=0$.
            We write, 
            \begin{align*}
                H(\tilde{\Pi}| E, Y=y)_{{\Psi}_{y}} & \geq H(\tilde{\Pi}|E,Y=y,\Gamma_y)_{{\Psi}_{y}}\\
                & = \Pr(\Gamma_y=0)H(\tilde{\Pi}|E,Y=y,\Gamma_y=0)_{{\Psi}_{y}} + \Pr(\Gamma_y=1)H(\tilde{\Pi}|E,Y=y,\Gamma_y=1)_{{\Psi}_{y}} \\
                & \geq \Pr(\Gamma_y=0)H(\tilde{\Pi}|E,Y=y,\Gamma_y=0)_{{\Psi}_{y}} \; .
            \end{align*}
            This allows us to use Equation~\eqref{eq:good_subspace_bound} and Equation~\eqref{eq:good_subspace_entropy} to bound each term, respectively, on the right hand side of the last inequality with
            \begin{align*}
                \begin{split}
                    H(\tilde{\Pi}| E,Y=y)_{{\Psi}_{y}}
                        & \geq
                            \qty
                                (
                                    1 - A(c)\qty(\sqrt{2}\sqrt[4]{1-\omega_p^y} +\sqrt{1-\omega_m^y}) - \eta(\lambda)
                                )
                    \\
                        & \phantom{\geq ~}
                            \times \qty
                                (
                                    \log_2 (1/c) - h\qty(\text{Pr}(\tilde{M}_{y} = 0) + \text{Pr}(\Gamma_y = 0) - 1)
                                )
                        \; .
                \end{split}
            \end{align*}
            We use Equation~\eqref{eq:good_subspace_bound} a second time together with Equation~\eqref{eq:mw_bound} to lower bound both probability terms in the argument of the binary entropy function and get
            \begin{align*}
                \begin{split}
                    H(\tilde{\Pi}|E,Y=y)_{{\Psi}_y}
                        & \geq
                            \qty
                                (
                                    1 - A(c)\qty(\sqrt{2}\sqrt[4]{1-\omega_p^y} +\sqrt{1-\omega_m^y}) - \eta(\lambda)
                                )
                            \times \\
                        & \phantom{\leq\leq} \qty(\log_{2}\qty(1/c) - h\qty(\omega_m^y - 2\sqrt{1-\omega_p^y}-\sqrt{2}A(c)\sqrt[4]{1-\omega_p^y}-A(c)\sqrt{1-\omega_m^y}-\eta(\lambda))) \; .
                \end{split}
            \end{align*}
            \item Now, taking the expectation over~$y$ on both sides of the inequality and using Lemma~\ref{lemma:jensen_extension}:
            \begin{align}\label{eq:psi_tilde_bound}
                \begin{split}
                    H(\tilde{\Pi}|E,Y)_{{\Psi}}
                    & \geq \qty(1-\sqrt{2}A(c)\sqrt[4]{1-\omega_p}-A(c)\sqrt{1-\omega_m}-\eta(\lambda)) \times \\
                    & \phantom{\leq ~} \qty(\log_{2}\qty(1/c) - h\qty(\omega_m - 2\sqrt{1-\omega_p}-\sqrt{2}A(c)\sqrt[4]{1-\omega_p}-A(c)\sqrt{1-\omega_m}-\eta(\lambda))) \; .
                \end{split}
            \end{align}
            \item Using Equation~\eqref{eq:binary_entropy_egligible_function_arg}, we can extract~$\eta(\lambda)$ from the argument of the binary entropy function in Equation~\eqref{eq:psi_tilde_bound} such that for some negligible function~$\xi(\lambda)$,
            \begin{align}\label{eq:psi_tilde_bound_negligible_outside}
                \begin{split}
                    H(\tilde{\Pi}|E,Y)_{{\Psi}}
                    & \geq \qty(1-\sqrt{2}A(c)\sqrt[4]{1-\omega_p}-A(c)\sqrt{1-\omega_m}) \times \\
                    & \phantom{\leq ~} \qty(\log_{2}\qty(1/c) - h\qty(\omega_m - 2\sqrt{1-\omega_p}-\sqrt{2}A(c)\sqrt[4]{1-\omega_p}-A(c)\sqrt{1-\omega_m}))-\xi(\lambda) \; .
                \end{split}
            \end{align}
            \item Using the continuity bound in Equation~\eqref{lemma:entropy_continuity} with~$\|{\Psi}-{\Phi}\|_{1} /2 \leq \sqrt{1-\omega_p}$ yields
                \begin{equation}\label{eq:psi_phi_continuity}
                    |H(\tilde{\Pi}|E,Y)_{{\Psi}} - H(\tilde{\Pi}|E,Y)_{{\Phi}}| \leq
                    \sqrt{1-\omega_p} \log_2 (3) + (1+\sqrt{1-\omega_p}) h\qty(\frac{\sqrt{1-\omega_p}}{1+\sqrt{1-\omega_p}}) \; .
                \end{equation}
            \item Combining Equation~\eqref{eq:psi_tilde_bound_negligible_outside} with Equation~\eqref{eq:psi_phi_continuity}, we conclude that the lemma holds. \qedhere
        \end{enumerate}
    \end{proof}    
    
    For the sake of brevity, denote the bound in Equation~\eqref{eq:simplified_entropy_lower_bound} as 
    \begin{equation*}
        H(\tilde{\Pi}|Y,E) \geq g(\omega_p, \omega_m, c) - \xi(\lambda) \; .
    \end{equation*}
    Seeing this inequality holds for all values~$c\in(1/2,1]$, for each winning probability pair~$(\omega_p, \omega_m)$ we can pick an optimal value of~$c$ to maximize the inequality. We do this implicitly and rewrite the bound as
    \begin{equation}\label{eq:two_variable_bound}
        g(\omega_p, \omega_m) \coloneqq \max_{c\in(1/2,1]} g(\omega_p, \omega_m, c)
    \end{equation}
    Doing so yields the graph in Figure~\ref{sub_fig:von_neumann_2}.

    In the protocol, later on, the verifier chooses whether to abort or not, depending on the overall winning probability~$\omega$:
    \begin{align*}
        \begin{split}
            \omega&\coloneqq\Pr(W=1) \\ &
            =\Pr(T=0)\Pr(W=1|T=0) + \Pr(T=1)\Pr(W=1|T=1)\\ &
            =  \underset{1-\beta}{\underbrace{\Pr(T=0)}}\cdot\omega_p +
            \underset{\beta}{\underbrace{\Pr(T=1)}}\cdot\omega_m
        \end{split}
    \end{align*}
    We therefore define for every~$\beta\in(0,1)$, another bound on the entropy, that depends only on $\omega$ (assuming~$\omega\geq1/2$) as
    \begin{equation}\label{eq:single_variable_bound}
        g(\omega; \beta)
        \coloneqq \min_{\omega_p}
        g\qty
        (
            \omega_p, \omega_m = \frac{\omega - (1-\beta)\omega_p}{\beta}
        ) \; .
    \end{equation}
    The optimal~$\beta$ can be found numerically. We plot the bound in Equation~\eqref{eq:single_variable_bound} in Figure~\ref{sub_fig:von_neumann} (Neglecting the negligible element~$\xi(\lambda)$).

\subsection{Entropy accumulation}\label{subsec:entropy_acc_comp}
    Combining Lemma~\ref{lemma:entropy_bound_proof} with Corollary~\ref{cor:entropy_reduction}, we now hold a lower bound for the von Neumann entropy of Protocol~\ref{protocol:simplified_single_round}--
    we can connect any winning probability~$\omega$ to a lower bound on the entropy of the pre-image test.
    This allows us to proceed with the task of entropy accumulation.
    To lower bound the total amount of smooth min-entropy accumulated throughout the entire execution of Protocol~\ref{protocol:simplified_accumulation}, we use the Entropy Accumulation Theorem (EAT), stated as Theorem~\ref{theorem:eat}.

    To use the EAT, we  need to first define the channels corresponding to Protocol~\ref{protocol:simplified_accumulation} followed by a proof that they are in fact EAT channels.
    In the notation of Definition~\ref{def:eat_channels}, we make the following choice of channels:
    \begin{equation*}
        \mathcal{M}_i:R_{i-1} \rightarrow
        R_i ~
        \underset{O_i}{\underbrace{\hat{\Pi}_i \hat{M}_i}} ~
        \underset{S_i}{\underbrace{K_i T_i G_i}} ~
        %\underset{Q_i}{\underbrace{W_i}}
    \end{equation*}
    and set $Q_i=W_i$.
    \begin{lemma}
        The channels~$\qty{\mathcal{M}_i:R_{i-1} \rightarrow R_i \hat{\Pi}_i \hat{M}_i K_i T_i G_i}_{i\in [n]}$ defined by the CPTP map describing the i-th round of Protocol~\ref{protocol:simplified_accumulation} as implemented by the computationally bounded untrusted device~$\tilde{D}$ and the verifier are EAT channels according to Definition~\ref{def:eat_channels}.
    \end{lemma}
    
    \begin{proof}
        To prove that the constructed channels~$\qty{\mathcal{M}_{i}}_{i\in [n]}$ are EAT channels we need to show that the conditions are fulfilled:
        \begin{enumerate}
            \item $\qty{O_i}_{i\in[n]}=\qty{\hat{\Pi}_i \hat{M}_i}_{i\in[n]}$,~$\qty{S_i}_{i\in[n]}=\qty{K_i T_i G_i}_{i\in[n]}$ and~$\qty{Q_i}_{i\in[n]}=\qty{W_i}_{i\in[n]}$ are all finite-dimensional classical systems. $\qty{R_i}_{i\in[n]}$ are arbitrary quantum systems. Finally, we have
            \begin{equation*}
                d_O =
                d_{\hat{\Pi}_i} \cdot d_{\hat{M}_i} =
                \qty|\qty{0,1,2}| \cdot \qty|\qty{0,1}| = 6 < \infty \;.
            \end{equation*}
            \item For any~$i\in[n]$ and any input state~$\sigma_{R_{i-1}}$, $W_i$ is a function of the classical values~$ \hat{\Pi}_i \hat{M}_i K_i T_i G_i$. Hence, the marginal~$\sigma_{O_i S_i}$ is unchanged when deriving~$W_i$ from it.
            \item For any initial state~$\rho^{in}_{R_0 E}$ and the resulting final state~$\rho_{_{\mathbf{OSQ}E}}=\rho_{_{\mathbf{\Pi M K T G W}E}}$ the Markov-chain conditions
            \begin{equation*}
                {(\hat{\Pi} \hat{M})}_{1}, \cdots ,{(\hat{\Pi} \hat{M})}_{i-1} \leftrightarrow {(K T G)}_{1}, \cdots ,{(K T G)}_{i-1}, E \leftrightarrow {(K T G)}_{i}
            \end{equation*}
            trivially hold for all~$i\in [n]$ as~$K_i,T_i$ and~$G_i$ are chosen independently from everything else.\footnote{For a reader interested in randomness expansion protocols, note that in order to use less randomness as an initial resource one could reuse the keys $K_i$ in some of the rounds (similarly to what happens in standard DI randomness expansion protocols). The Markov-chain condition also holds when the keys are being reused and so one can still follow our proof technique.} \qedhere
        \end{enumerate}
    \end{proof}

    \begin{lemma}
        Let~$\beta,\omega,\gamma\in(0,1)$ and $g(\omega;\beta)$ be the function in Equation~\eqref{eq:single_variable_bound}.
        Let~$p$ be a probability distribution over~$\mathcal{W}=\qty{\perp, 0, 1}$ such that~$\gamma=1-p(\perp)$ and $\omega=p(1)/\gamma$. Define~$\Sigma(p) = \qty{\sigma_{R_{i-1}R'} : \mathcal{M}_i(\sigma)_{W_i} = p}$.
        Then, there exists a negligible function~$\xi(\lambda)$ such that
        \begin{equation}\label{eq:comp_min_tradeoff}
            (g(\omega; \beta)-\xi(\lambda)) (1 - \beta\gamma) \leq
            \inf_{\sigma_{R_{i-1} R'} \in \Sigma(p)}
            H{(\hat{\Pi}_i \hat{M}_i | K_i T_i G_i R')}_{\mathcal{M}_i (\sigma)} \; .
        \end{equation}
        In particular, this implies that for every~$\beta\in(0,1)$ the function
        \begin{equation}\label{eq:f_min}
                f_{\text{min}}(p) \coloneqq \qty(g(\omega; \beta) - \xi(\lambda)) \qty(1 - \beta\gamma)
        \end{equation}
            satisfies Definition~\ref{def:min_tradeoff_function} and is therefore a min-tradeoff function.
    \end{lemma}

    \begin{proof}
        Due to Lemma~\ref{lemma:entropy_bound_proof} and the consequent Equation~\eqref{eq:single_variable_bound}, the following holds for any polynomial sized state~$\sigma$ (not necessarily efficient):
        \begin{align*}
            g(\omega; \beta) - \xi(\lambda) & \leq H{(\hat{\Pi}_i | Y_i K_i, T_i = 0, R')}_{\mathcal{M}_i (\sigma)} \\
            % \begin{split}
            %     & = \frac{1}{\Pr(T_i=0)} \Pr(T_i=0) H{(\Pi_i | Y_i K_i, T_i = 0, R')}_{\mathcal{M}_i (\sigma)}
            % \end{split}
            % \\
            % \begin{split}
            %     & \leq \frac{1}{\Pr(T_i=0)} \Pr(T_i=0) H{(\Pi_i | Y_i K_i, T_i = 0, R')}_{\mathcal{M}_i (\sigma)}  \\
            %     & ~~~~ + \frac{1}{\Pr(T_i=0)} \Pr(T_i=1)H{(\Pi_i | Y_i K_i, T_i = 1, R')}_{\mathcal{M}_i (\sigma)}
            % \end{split}
            % \\
            \begin{split}
                & \leq \frac{1}{\Pr(T_i=0)} \left[\Pr(T_i=0) H{(\hat{\Pi}_i | Y_i K_i, T_i = 0, R')}_{\mathcal{M}_i (\sigma)}  \right. \\
                & \phantom{= \frac{1}{\Pr(T_i=0)} ~~} + \left. \Pr(T_i=1)H{(\hat{\Pi}_i | Y_i K_i, T_i = 1, R')}_{\mathcal{M}_i (\sigma)}\right]
            \end{split}
            \\
            & = \frac{1}{\Pr(T_i=0)} H{(\hat{\Pi}_i | Y_i K_i T_i R')}_{\mathcal{M}_i (\sigma)} \\
            & = ~~~~ \frac{1}{1 - \beta\gamma} H{(\hat{\Pi}_i | Y_i K_i T_i R')}_{\mathcal{M}_i (\sigma)} \\
            & \leq ~~~~ \frac{1}{1 - \beta\gamma} H{(\hat{\Pi}_i \hat{M}_i | K_i T_i R')}_{\mathcal{M}_i (\sigma)} \\
            &  = ~~~~ \frac{1}{1 - \beta\gamma} H{(\hat{\Pi}_i \hat{M}_i | K_i T_i G_i R')}_{\mathcal{M}_i (\sigma)}\;.
        \end{align*}
        For the last equality, note that the device only knows which test to perform while being unaware if it is for a generation round or not.
        Therefore, once~$T$ is given,~$G$ does not provide any additional information. \qedhere
    \end{proof}

    Using Theorem~\ref{theorem:eat}, we can bound the smooth min-entropy resulting in application of Protocol~\ref{protocol:simplified_accumulation}.
    \begin{equation}\label{eq:fin}
        \hmin{\varepsilon_s} \qty(\mathbf{\hat{\Pi} \hat{M}}| \mathbf{KTG})_{\rho _{| \Omega}} 
        \geq
        n f_{\text{min}}
        - \mu \sqrt{n}
        \;,
    \end{equation}
    where~$f_{\text{min}}$ is given by Equation~\eqref{eq:f_min}.
    We simplify the right-hand side in Equation~\eqref{eq:fin} to a single entropy accumulation rate,~$\mu_\mathrm{opt}$, and a negligible reduction~$\xi(\lambda)$.
    \begin{equation}\label{eq:mu_opt_def}
        \hmin{\varepsilon}
        \qty(\mathbf{\hat{\Pi} \hat{M}}| K T G)_{\rho _{| \Omega}}
        \geq
        n (\mu_{\text{opt}} (n, \omega, \gamma, \varepsilon_s, p_\Omega; \beta) 
        - \xi(\lambda) ) \; .
    \end{equation}
    Note that the differential of~$f_{\mathrm{min}}$ is unbounded.
    This prevents us from using the EAT due to Equation~\eqref{eq:second_order_term}.
    This issue is addressed by defining a new min-tradeoff function~$\tilde{f}_{\mathrm{min}}$ such that
    \begin{equation}
        \tilde{f}_\mathrm{min} (\omega; \omega_0) =
        \begin{cases}
            f_\mathrm{min} (\omega) & \omega \leq \omega_0 \\
            \frac{d}{d\omega}\left.f_\mathrm{min} (\omega)\right|_{\omega=\omega_0} (\omega - \omega_0) + f_\mathrm{min} (\omega_0) & \mathrm{otherwise}
        \end{cases}
        \; .
    \end{equation}
    We provide a number of plots for~$\mu_\text{opt}$ as a function of~$\omega$ for various values of~$n$ in Figure~\ref{fig:mu_opt_plots}. We remark that we did not fully optimize the code to derive the plots and one can probably derive tighter plots.

\subsection{Randomness rates}\label{sec:plots}
    In the previous section we derived a lower bound on $\hmin{\varepsilon}
        (\mathbf{\hat{\Pi} \hat{M}}| K T G)_{\rho _{| \Omega}}$. 
    One is, however, interested in a bound on $  \hmin{\varepsilon_s}(\mathbf{\hat{\Pi}}|\mathbf{KTG}E)_{\rho_{|\Omega}} $ instead. 
    To derive the desired bound we follow similar steps to those taken in the proof of~\cite[Lemma 11.8]{arnon2020device}.
    
    \begin{theorem}\label{lemma:entropy_isolation}
    Under the same conditions for deriving Equation~\eqref{eq:fin}, the following holds:
        \begin{align*}
            \begin{split}
                \hmin{\varepsilon_s}(\mathbf{\hat{\Pi}}|\mathbf{KTG}E)_{\rho_{|\Omega}} 
                & \geq
                n \qty(\mu_{\text{opt}} (n, \omega, \gamma, \varepsilon_s/4, p_\Omega; \beta)
                - \xi(\lambda) - \gamma) \\
                & \phantom{\leq ~ }
                - 2\log(7) \sqrt{1 - 2 \log(\frac{\varepsilon_s}{4} \cdot p_{\Omega} )}- 3\log(1 - \sqrt{1-{(\varepsilon_s/4)}^2}) \;.
            \end{split}
        \end{align*}
    \end{theorem}

    \begin{proof}
        We begin with entropy chain rule~\cite[Theorem 6.1]{tomamichel2016}
        \begin{align*}
            \begin{split}
                \hmin{\varepsilon_s}(\mathbf{\hat{\Pi}}|\mathbf{KTG}E)_{\rho_{|\Omega}} & 
                \geq \hmin{\frac{\varepsilon_s}{4}}(\mathbf{\hat{\Pi} \hat{M}}|\mathbf{KTG}E)_{\rho_{|\Omega}}
                - \hmax{\frac{\varepsilon_s}{4}}(\mathbf{\hat{M}}|\mathbf{KTG}E)_{\rho_{|\Omega}} \\
                & ~~~ - 3\log(1 - \sqrt{1-{(\varepsilon_s/4)}^2})    \;.
            \end{split}
        \end{align*}
        The first term on the right-hand side is given in Equation~\eqref{eq:mu_opt_def}; it remains to find an upper bound for the second term. Let us start from
        \begin{equation*}
            \hmax{\frac{\varepsilon_s}{4}}(\mathbf{\hat{M}}|\mathbf{KGT}E)_{\rho_{|\Omega}}
            \leq \hmax{\frac{\varepsilon_s}{4}}(\mathbf{\hat{M}}|\mathbf{T}E)_{\rho_{|\Omega}} \;.
        \end{equation*}
        We then use the EAT again in order to bound $\hmax{\frac{\varepsilon_s}{4}}(\mathbf{\hat{M}}|\mathbf{T}E)_{\rho_{|\Omega}}$.
        We identify the EAT channels with~$\mathbf{O}\rightarrow\mathbf{\hat{M}},\mathbf{S}\rightarrow\mathbf{T}$ and~$E \rightarrow E$.
        The Markov conditions then trivially hold and the max-tradeoff function reads
        \begin{equation*}
            f_{\text{max}}(p) \geq
            \sup_{\sigma_{R_{i-1}R'}:\mathcal{M}_{i}(\sigma)_{W_i}=\omega}
            H(\hat{M}_i|T_i R')_{\mathcal{M}_i (\sigma)}\;.
        \end{equation*}
        Since the following distributions are satisfied for all~$i\in\qty[n]$
        \begin{align*}
            \Pr[\hat{M}_i =\perp|T_i =0] & = 1, \\
            \Pr[\hat{M}_i \in \qty{0,1}{|T_i =1}] & = 1, \\
            \Pr[T_i =1] & = \gamma,
        \end{align*}
        the max-tradeoff function is simply~$f_{\text{max}}(p)=\gamma$ (thus $\left\Vert \nabla f_{\text{max}} \right\Vert _{\infty} = 0)$.
        We therefore get
        \begin{equation*}
            \hmax{\frac{\varepsilon_s}{4}}
            \qty(\mathbf{\hat{M}}|\mathbf{T}E)
            _{\rho_{|\Omega}}
            \leq \gamma n + 2 \log (7)
            \sqrt{1 - 2 \log(\frac{\varepsilon_s}{4} \cdot p_{\Omega} )} \; .
            \qedhere
        \end{equation*}
    \end{proof}

        We wish to make a final remark. One could repeat the above analysis under the assumption that~$Y$ and~$X$ are leaked, to get a stronger security statement (this is however not done in previous works). In this case, $Y_i$ and $X_i$ should be part of the output $O_i$ of the channel.
        Then, the dimension of~$O_i$ scales exponentially with~$\lambda$, which worsens the accumulation rate~$\mu_{\text{opt}}$ since this results in exponential factors in Equation~\eqref{eq:eat}.
        In that case,~$\mu_{\text{opt}}$ becomes a monotonically decreasing function of~$\lambda$ and we get a certain tradeoff between the two elements in the following expression:
        \begin{equation*}
            \mu_{\text{opt}}(n, \omega, \gamma, \varepsilon_s, p_\Omega, \lambda; \beta) - \xi(\lambda) \; .
        \end{equation*}
        This means that there is an optimal value of~$\lambda$ that maximizes the entropy and to find it one requires an explicit bound on $\xi(\lambda)$.

    \begin{figure}
        \centering
        \includegraphics[scale=0.6]{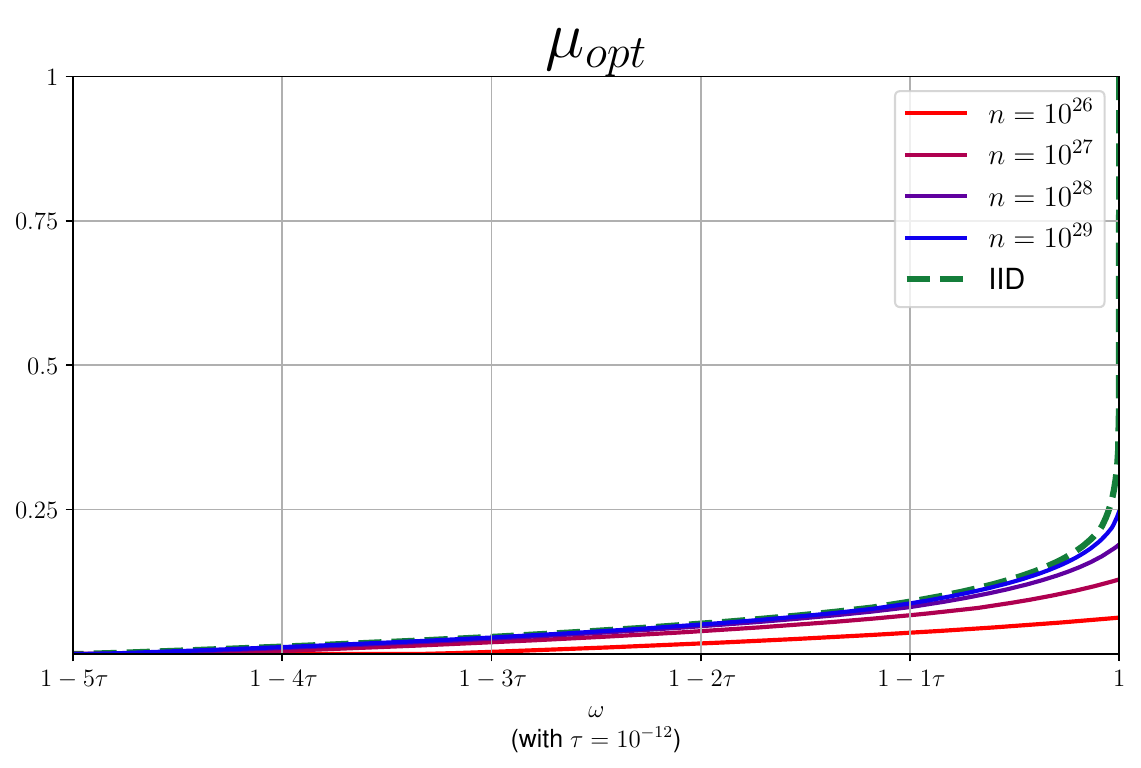}
        \caption
        {
            Plots of the entropy accumulation rates $\mu_{\text{opt}}(n, \omega, \gamma=1, \varepsilon_s = {10}^{-5}, p_\Omega = {10}^{-5})$ as a function of~$\omega$ for various values of~$n$. For reference, we include the accumulation rate of the IID case as a function of~$\omega$ (Fig~\ref{sub_fig:von_neumann}), to which all rates converge in the limit~$n\rightarrow\infty$.
            We neglect the negligible element in Equation~\eqref{eq:mu_opt_def}.
        }
        \label{fig:mu_opt_plots}
    \end{figure}   

    \section{Conclusion and outlook}\label{sec:conc}
    By utilizing a combination of results from quantum information theory and post-quantum cryptography, we have shown that entropy accumulation is feasible when interacting with a single device.
While this was previously done in~\cite{brakerski2021cryptographic} using ad hoc techniques, we provide a flexible framework that builds on well-studied tools and follows similar steps to those used in DI protocols based on the violation of Bell inequalities using two devices~\cite{arnon2019simple}. Prior to our work, it was believed that such an approach cannot be taken (see the discussion in~\cite{brakerski2021cryptographic}).
We remark that while we focused on randomness \emph{certification} in the current manuscript, one could now easily extend the analysis to randomness \emph{expansion, amplification} and \emph{key distribution} using the same standard techniques applied when working with two devices~\cite{arnon2019simple,kessler2020device}.
%The main difficulty was proving the relation between the von Neumann entropy generated by a box in an IID assumption, with its winning probability while showing that the entropy accumulates was relatively straight forward due to existing tools in quantum information theory.

Furthermore, even though we carried out the proof here specifically for the computational challenge derived from a  NTCF, the methods that we establish are modular and can be generalized to other protocols with different cryptographic assumptions. 
For example, in two recent works~\cite{natarajan2023bounding,brakerski2023simple}, the winning probability in various ``computational games'' are  tied to the anti-commutator of the measurements used by the device that plays the game (see~\cite[Lemma 38]{natarajan2023bounding} and~\cite[Theorem 4.7]{brakerski2023simple} in particular). Thus, their results can be used to derive a bound on the conditional von Neumann entropy as we do in Lemma~\ref{lemma:entropy_bound_proof}.
From there onward, the final bound on the accumulated smooth min-entropy is derived exactly as in our work.

Apart from the theoretical contribution, the new proof method allows us to derive explicit bounds for a finite number of rounds of the protocol, in contrast to asymptotic statements. Thus, one can use the bounds to study the practicality of DI randomness certification protocols based on computational assumptions. 

For the \emph{current protocol}, in order to get a positive rate the number of repetitions~$n$ required, as seen in Figure~\ref{fig:mu_opt_plots}, is too demanding for actual implementation. In addition, the necessary observed winning probability is extremely high.
We pinpoint the ``source of the problem'' to the min-tradeoff function presented in Figure~\ref{fig:pi_von_neumann_both} and provide below a number of suggestions as to how one might improve the derived bounds. 
In general, however, we expect that for other protocols (e.g., those suggested in~\cite{natarajan2023bounding,brakerski2023simple}) one could derive better min-tradeoff functions that will bring us closer to the regime of experimentally relevant protocols. 
In fact, our framework allows us to compare different protocols via their min-tradeoff functions and thus can be used as a tool for bench-marking new protocols. 

We conclude with several open questions.
\begin{enumerate}
    \item In both the original analysis done in~\cite{brakerski2021cryptographic} and our work, the cryptographic assumption needs to hold even when the (efficient) device is getting an (inefficient) ``advice state''. Including the advice states is necessary when using the entropy accumulation in its current form, due to the usage of a min-tradeoff function (see Equation~\eqref{eq:comp_min_tradeoff}). 
    One fundamental question is therefore whether this is \emph{necessary} in all DI protocols based on post-quantum computational assumptions or not. 
    \item As mentioned above, what makes the protocol considered here potentially unfeasible for experimental implementations is its min-tradeoff function. Ideally, one would like to both decrease the winning probability needed in order to certify entropy as well as the \emph{derivative} of the function, which is currently too large. The derivative impacts the second-order term of the accumulated entropy and this is why we observe the need for a large number of rounds in the protocol-- many orders of magnitude more than in the DI setup with two devices. 
    Any improvement of Lemma~\ref{lemma:entropy_bound_proof} may be useful; when looking into the details of the proof, there is indeed some room for it. 
    \item Once one is interested in non-asymptotic statements and actual implementations, the unknown negligible function~$\xi(\lambda)$ needs to be better understood. The assumption is that $\xi(\lambda)=0$ as $\lambda\rightarrow\inf$ but in any given execution one does fix a finite $\lambda$. Some more explicit statements should then be made regarding $\xi(\lambda)$ and incorporated into the final bound. 
    \item In the current manuscript we worked with the EAT presented in~\cite{dupuis2020entropy}. A generalized version, that allows for potentially more complex protocols, appears in~\cite{metger2022generalised}. All of our lemmas and theorems can also be derived using~\cite{metger2022generalised} without any modifications. An interesting question is whether there are DI protocols with a single device that can exploit the more general structure of~\cite{metger2022generalised}. 
\end{enumerate}

    \begin{acknowledgments}
    The authors would like to thank Zvika Brakerski, Tony Metger, Thomas Vidick and Tina Zhang for useful discussions.
    This research was generously supported by the Peter and Patricia Gruber Award, the Daniel E. Koshland Career Development Chair, the Koshland Research Fund, the Karen Siem Fellowship for Women in Science and the Israel Science Foundation (ISF) and the Directorate for Defense Research and Development (DDR\&D), grant No. 3426/21.
\end{acknowledgments}

\bibliographystyle{alpha}
\bibliography{main.bbl}

\newcommand{\etalchar}[1]{$^{#1}$}
\begin{thebibliography}{BGKM{\etalchar{+}}23}

\bibitem[AF20]{arnon2020device}
Rotem Arnon-Friedman.
\newblock {\em Device-Independent Quantum Information Processing: A Simplified
  Analysis}.
\newblock Springer Nature, 2020.

\bibitem[AFDF{\etalchar{+}}18]{arnon2018practical}
Rotem Arnon-Friedman, Fr{\'e}d{\'e}ric Dupuis, Omar Fawzi, Renato Renner, and
  Thomas Vidick.
\newblock Practical device-independent quantum cryptography via entropy
  accumulation.
\newblock {\em Nature communications}, 9(1):459, 2018.

\bibitem[AFRV19]{arnon2019simple}
Rotem Arnon-Friedman, Renato Renner, and Thomas Vidick.
\newblock Simple and tight device-independent security proofs.
\newblock {\em SIAM Journal on Computing}, 48(1):181--225, 2019.

\bibitem[AH23]{aaronson2023certified}
Scott Aaronson and Shih-Han Hung.
\newblock Certified randomness from quantum supremacy, 2023.

\bibitem[BCC{\etalchar{+}}10]{berta2010uncertainty}
Mario Berta, Matthias Christandl, Roger Colbeck, Joseph~M Renes, and Renato
  Renner.
\newblock The uncertainty principle in the presence of quantum memory.
\newblock {\em Nature Physics}, 6(9):659--662, 2010.

\bibitem[BCM{\etalchar{+}}21]{brakerski2021cryptographic}
Zvika Brakerski, Paul Christiano, Urmila Mahadev, Umesh Vazirani, and Thomas
  Vidick.
\newblock A cryptographic test of quantumness and certifiable randomness from a
  single quantum device.
\newblock {\em Journal of the ACM (JACM)}, 68(5):1--47, 2021.

\bibitem[BCP{\etalchar{+}}14]{brunner2014bell}
Nicolas Brunner, Daniel Cavalcanti, Stefano Pironio, Valerio Scarani, and
  Stephanie Wehner.
\newblock Bell nonlocality.
\newblock {\em Reviews of modern physics}, 86(2):419, 2014.

\bibitem[BGKM{\etalchar{+}}23]{brakerski2023simple}
Zvika Brakerski, Alexandru Gheorghiu, Gregory~D. Kahanamoku-Meyer, Eitan Porat,
  and Thomas Vidick.
\newblock Simple tests of quantumness also certify qubits.
\newblock {\em arXiv preprint}, 2023.

\bibitem[BKVV20]{brakerski2020simpler}
Zvika Brakerski, Venkata Koppula, Umesh Vazirani, and Thomas Vidick.
\newblock Simpler proofs of quantumness.
\newblock {\em arXiv preprint arXiv:2005.04826}, 2020.

\bibitem[DF19]{dupuis2019entropy}
Fr{\'e}d{\'e}ric Dupuis and Omar Fawzi.
\newblock Entropy accumulation with improved second-order term.
\newblock {\em IEEE Transactions on information theory}, 65(11):7596--7612,
  2019.

\bibitem[DFR20]{dupuis2020entropy}
Frederic Dupuis, Omar Fawzi, and Renato Renner.
\newblock Entropy accumulation.
\newblock {\em Communications in Mathematical Physics}, 379(3):867--913, 2020.

\bibitem[DPVR12]{de2012trevisan}
Anindya De, Christopher Portmann, Thomas Vidick, and Renato Renner.
\newblock Trevisan's extractor in the presence of quantum side information.
\newblock {\em SIAM Journal on Computing}, 41(4):915--940, 2012.

\bibitem[GMP22]{gheorghiu2022quantum}
Alexandru Gheorghiu, Tony Metger, and Alexander Poremba.
\newblock Quantum cryptography with classical communication: parallel remote
  state preparation for copy-protection, verification, and more.
\newblock {\em arXiv preprint arXiv:2201.13445}, 2022.

\bibitem[GV19]{gheorghiu2019computationally}
Alexandru Gheorghiu and Thomas Vidick.
\newblock Computationally-secure and composable remote state preparation.
\newblock In {\em 2019 IEEE 60th Annual Symposium on Foundations of Computer
  Science (FOCS)}, pages 1024--1033. IEEE, 2019.

\bibitem[KAF20]{kessler2020device}
Max Kessler and Rotem Arnon-Friedman.
\newblock Device-independent randomness amplification and privatization.
\newblock {\em IEEE Journal on Selected Areas in Information Theory},
  1(2):568--584, 2020.

\bibitem[KMCVY22]{kahanamoku2022classically}
Gregory~D Kahanamoku-Meyer, Soonwon Choi, Umesh~V Vazirani, and Norman~Y Yao.
\newblock Classically verifiable quantum advantage from a computational bell
  test.
\newblock {\em Nature Physics}, 18(8):918--924, 2022.

\bibitem[LG22]{liu2022depth}
Zhenning Liu and Alexandru Gheorghiu.
\newblock Depth-efficient proofs of quantumness.
\newblock {\em Quantum}, 6:807, 2022.

\bibitem[Mah18]{mahadev2018classical}
Urmila Mahadev.
\newblock Classical verification of quantum computations.
\newblock In {\em 2018 IEEE 59th Annual Symposium on Foundations of Computer
  Science (FOCS)}, pages 259--267. IEEE, 2018.

\bibitem[MDCAF21]{metger2021device}
Tony Metger, Yfke Dulek, Andrea Coladangelo, and Rotem Arnon-Friedman.
\newblock Device-independent quantum key distribution from computational
  assumptions.
\newblock {\em New Journal of Physics}, 23(12):123021, 2021.

\bibitem[MFSR22]{metger2022generalised}
Tony Metger, Omar Fawzi, David Sutter, and Renato Renner.
\newblock Generalised entropy accumulation.
\newblock In {\em 2022 IEEE 63rd Annual Symposium on Foundations of Computer
  Science (FOCS)}, pages 844--850. IEEE, 2022.

\bibitem[MV21]{metger2021self}
Tony Metger and Thomas Vidick.
\newblock Self-testing of a single quantum device under computational
  assumptions.
\newblock {\em Quantum}, 5:544, 2021.

\bibitem[NZ23]{natarajan2023bounding}
Anand Natarajan and Tina Zhang.
\newblock Bounding the quantum value of compiled nonlocal games: from {CHSH} to
  {BQP} verification.
\newblock {\em arXiv preprint arXiv:2303.01545}, 2023.

\bibitem[PAB{\etalchar{+}}09]{Pironio_2009}
Stefano Pironio, Antonio Ac{\'{\i} }n, Nicolas Brunner, Nicolas Gisin, Serge
  Massar, and Valerio Scarani.
\newblock Device-independent quantum key distribution secure against collective
  attacks.
\newblock {\em New Journal of Physics}, 11(4):045021, apr 2009.

\bibitem[PGT{\etalchar{+}}22]{primaatmaja2022security}
Ignatius~W Primaatmaja, Koon~Tong Goh, Ernest Y-Z Tan, John T-F Khoo, Shouvik
  Ghorai, and Charles C-W Lim.
\newblock Security of device-independent quantum key distribution protocols: a
  review.
\newblock {\em arXiv preprint arXiv:2206.04960}, 2022.

\bibitem[Reg09]{regev2009lattices}
Oded Regev.
\newblock On lattices, learning with errors, random linear codes, and
  cryptography.
\newblock {\em Journal of the ACM (JACM)}, 56(6):1--40, 2009.

\bibitem[RK05]{renner2005universally}
Renato Renner and Robert K{\"o}nig.
\newblock Universally composable privacy amplification against quantum
  adversaries.
\newblock In {\em Theory of Cryptography: Second Theory of Cryptography
  Conference, TCC 2005, Cambridge, MA, USA, February 10-12, 2005. Proceedings
  2}, pages 407--425. Springer, 2005.

\bibitem[Sca19]{scarani2019bell}
Valerio Scarani.
\newblock {\em Bell nonlocality}.
\newblock Oxford University Press, 2019.

\bibitem[SCR{\etalchar{+}}22]{stricker2022towards}
Roman Stricker, Jose Carrasco, Martin Ringbauer, Lukas Postler, Michael Meth,
  Claire Edmunds, Philipp Schindler, Rainer Blatt, Peter Zoller, Barbara Kraus,
  et~al.
\newblock Towards experimental classical verification of quantum computation.
\newblock {\em arXiv preprint arXiv:2203.07395}, 2022.

\bibitem[TCR10]{tomamichel2010duality}
Marco Tomamichel, Roger Colbeck, and Renato Renner.
\newblock Duality between smooth min-and max-entropies.
\newblock {\em IEEE Transactions on information theory}, 56(9):4674--4681,
  2010.

\bibitem[Tom16]{tomamichel2016}
Marco Tomamichel.
\newblock {\em Quantum Information Processing with Finite Resources}.
\newblock Springer International Publishing, 2016.

\bibitem[VZ21]{vidick2021classical}
Thomas Vidick and Tina Zhang.
\newblock Classical proofs of quantum knowledge.
\newblock In {\em Advances in Cryptology--EUROCRYPT 2021: 40th Annual
  International Conference on the Theory and Applications of Cryptographic
  Techniques, Zagreb, Croatia, October 17--21, 2021, Proceedings, Part II},
  pages 630--660. Springer, 2021.

\bibitem[Wil13]{wilde2013gentle}
Mark~M. Wilde.
\newblock {\em Quantum Information Theory}.
\newblock Cambridge University Press, 2013.

\bibitem[Win16]{winter2016continuity}
Andreas Winter.
\newblock Tight uniform continuity bounds for quantum entropies: Conditional
  entropy, relative entropy distance and energy constraints.
\newblock {\em Communications in Mathematical Physics}, 347(1):291--313, mar
  2016.

\bibitem[ZKML{\etalchar{+}}21]{zhu2021interactive}
Daiwei Zhu, Gregory~D Kahanamoku-Meyer, Laura Lewis, Crystal Noel, Or~Katz,
  Bahaa Harraz, Qingfeng Wang, Andrew Risinger, Lei Feng, Debopriyo Biswas,
  et~al.
\newblock Interactive protocols for classically-verifiable quantum advantage.
\newblock {\em arXiv preprint arXiv:2112.05156}, 2021.

\bibitem[ZvLAF{\etalchar{+}}23]{zapatero2023advances}
V{\'\i}ctor Zapatero, Tim van Leent, Rotem Arnon-Friedman, Wen-Zhao Liu, Qiang
  Zhang, Harald Weinfurter, and Marcos Curty.
\newblock Advances in device-independent quantum key distribution.
\newblock {\em npj Quantum Information}, 9(1):10, 2023.

\end{thebibliography}

\end{document}